\documentclass{article}

\usepackage[preprint]{neurips_2026}
\usepackage[normalem]{ulem}
\usepackage{enumitem}
\usepackage{float}
\usepackage{graphicx}
\usepackage{dcolumn}
\usepackage{amssymb}
\usepackage{appendix}
\usepackage{physics}
\usepackage{mathtools}
\usepackage{esvect}
\usepackage{wrapfig}
\usepackage{amsthm}
\usepackage{verbatim}
\usepackage{bbm}
\usepackage[colorlinks=true,linkcolor=blue,citecolor=red,plainpages=false,pdfpagelabels]{hyperref}
\usepackage{subfigure}

\usepackage{placeins}
\usepackage[mathscr]{euscript}

\usepackage{amsmath,amssymb,amsfonts}
\usepackage{algorithmic}
\usepackage{graphicx}
\usepackage{textcomp}
\usepackage[svgnames]{xcolor}
\usepackage{booktabs}
\usepackage{caption}
\usepackage{hyperref}
\usepackage{colortbl}
\usepackage{amsmath}
 \usepackage{bbm}
 \usepackage{braket}
\usepackage[normalem]{ulem}

\usepackage{amsthm, amssymb}
\usepackage{amsmath}

\let\Tr\relax

\definecolor{lightgreen}{HTML}{90EE90}
\def\Tr{\operatorname{Tr}}

\def \Hcal {\mathcal{H}}
\def \Kcal{\mathcal{K}}
\def \Ebb{\mathbb{E}}
\def \pbf{\mathbf{p}}

\def\({\left(}
\def\){\right)}
\def\[{\left[}
\def\]{\right]}
\def \thetavec {\vec{\theta}}
\def \Lcal {\mathcal{L}}
\def \Rcal{\mathcal{R}}

\let\emptyset\varnothing

\newtheorem{theorem}{Theorem}

\newtheorem{corollary}{Corollary}

\newtheorem{definition}{Definition}

\newtheorem{lemma}{Lemma}

\newtheorem{proposition}{Proposition}
\newtheorem{remark}{Remark}

\newtheorem{note}{Note}

\def\>{\rangle}
\def\<{\langle}

\usepackage[utf8]{inputenc} 
\usepackage[T1]{fontenc}    
\usepackage{hyperref}       
\usepackage{url}            
\usepackage{booktabs}       
\usepackage{amsfonts}       
\usepackage{nicefrac}       
\usepackage{microtype}      
\usepackage{xcolor}         

\title{Adversarial Effects on Expressibility and Trainability in Distributed Variational Quantum Algorithms}

\author{
  Abhishek Sadhu \\
  School of Computer Science, University of Birmingham \\
  Edgbaston, Birmingham B15 2TT, United Kingdom\\
  \texttt{a.sadhu@bham.ac.uk}\\
  \And
  Sharu Theresa Jose \\
  School of Computer Science, University of Birmingham \\
  Edgbaston, Birmingham B15 2TT, United Kingdom\\
  \texttt{s.t.jose@bham.ac.uk}
  }

\begin{document}

\maketitle

\begin{abstract}
    Distributed quantum algorithms offer a promising pathway to scale variational quantum algorithms beyond the constraints of noisy intermediate-scale quantum hardware. However, existing approaches implicitly assume a trusted entanglement-sharing layer across quantum processors. We show that this assumption introduces a fundamental vulnerability: adversarial perturbations of shared entanglement induce structured gate-level noise that directly impacts quantum learning.
    We develop a framework that maps entanglement-level perturbations to gate-level noise via an explicit Kraus representation. To quantify their impact, we introduce Kraus expressibility, a metric that generalizes unitary expressibility to noisy quantum channels. We then establish a trade-off between Kraus expressibility and trainability of noisy quantum circuits through gradient variance analysis.
  Our analysis reveals that an adversary can manipulate Kraus expressibility to maintain sufficiently large cost gradients (avoiding barren plateaus) while systematically biasing optimization toward incorrect solutions. We validate these findings through numerical simulations, demonstrating adversarial degradation of expressibility and trainability.  
\end{abstract}

\vspace{-0.1cm}
\section{Introduction}
\vspace{-0.1cm}
The realization of a desirably functional quantum machine learning (QML) framework~\cite{ventura1999quantum,biamonte2017quantum}  beyond the limitations of near-term hardware is a central challenge for the field. Despite steady experimental progress across different architectures~\cite{google2021weber,ibmquantum_computers}, currently available quantum processors are limited by error-corrected qubit count~\cite{vezvaee2026demonstration,hong2024entangling,reichardt2024fault}, short lifetime of physical qubits~\cite{ofek2016extending,wang2022towards}, and imperfect gate fidelities~\cite{jin2025superconducting,marxer2025above}. 

In this era of noisy quantum devices, variational quantum algorithms (VQAs)~\cite{mcclean2016theory,cerezo2021variational} have emerged as a promising framework for achieving practical quantum utility~\cite{bharti2022noisy}. 
VQAs typically employ parameterized quantum circuits (PQCs) trained via classical optimization loops~\cite{kundu2024kanqas,sadhu2024quantum,kundu2026reinforcement}, and have demonstrated early empirical success across tasks such as classification \cite{jager2023universal}, clustering~\cite{otterbach2017unsupervised}, and generative modeling \cite{hibat2024framework, huang2025generative}. 
 However, their scalability remains fundamentally constrained by the limitations of currently available noisy-intermediate scale quantum (NISQ)~\cite{preskill2018quantum} hardware.

\textit{Distributed quantum computing}~\cite{cirac1999distributed,diadamo2021distributed} offers a pathway to overcome these hardware limitations by 
partitioning a quantum circuit across multiple,  small-scale NISQ quantum processing units (QPUs),  thereby enabling the execution of circuits beyond the capacity of  single devices.  In distributed VQAs, 
this requires sharing quantum correlations (such as entanglement~\cite{horodecki2009quantum} or nonlocality~\cite{brunner2014bell,sadhu2023testing}) across QPUs to  implement non-local multi-qubit quantum gates across them.
However, this architecture introduces new vulnerabilities. 
In practice, entanglement is distributed via quantum repeater networks~\cite{briegel1998quantum,azuma2023quantum,sadhu2023practical}, creating new attack surfaces  \cite{tunc2025resilience, harkness2025security}. For example, entangled photon sources~\cite{kaur2025chip,chen2021quantum,liu2019solid} (e.g., SPDC~\cite{couteau2018spontaneous} based network nodes) can be compromised through pump manipulation or noise injection, degrading the quality of shared entanglement. 

Such adversarial perturbations to the shared entanglement can corrupt non-local gate execution, leading to cascading errors across distributed quantum circuits. From a machine learning perspective, this vulnerability is particularly critical. VQAs rely on precise circuit evaluations to estimate gradients and optimize model parameters. The gate noise introduced via the entanglement-level perturbations directly affect circuit expressibility, and optimization dynamics. \textit{Despite this, the effect of entanglement-level attacks on distributed quantum learning remains largely unexplored}.

\emph{\textbf{Contributions}:} We address this gap by introducing a novel adversarial framework for distributed VQAs, where the attacker perturbs the shared entangled qubit pair underlying inter-device communication. We focus on the \textit{cat-entangler/cat-disentangler} protocol \cite{yimsiriwattana2004generalized},  which implements distributed, non-local CNOT gates using a single entangled qubit pair and two bits of classical communication. As a minimal and widely studied primitive for distributed quantum computation, it provides a natural setting for adversarial analysis. 
Our main contributions are as follows:

(1) We show that perturbations to shared entanglement between QPUs systematically corrupt non-local multi-qubit gates,  transforming ideal unitary gates into \textit{noisy quantum channels}. We derive an explicit Kraus representation of the resulting channel, parameterized by the adversary, providing a principled mapping from entanglement-level attacks to gate-level noise.

(2) We introduce \textit{Kraus expressibility}, a novel expressibility measure that quantifies how well `noisy' quantum circuits can uniformly explore the space of  ideal unitaries. This generalizes \textit{unitary expressibility} for noise-less PQCs \cite{holmes2022connecting} to account for both adversarial and intrinsic hardware noise. We show that adversarial perturbations  degrade Kraus expressibility through two mechanisms: $(i)$ state purity degradation and $(ii)$ noise-induced cross-correlations. Our analysis introduces a framework to evaluate vulnerability of PQCs to stealthy noise manipulation.

(3) We establish a theoretical link between Kraus expressibility and \textit{trainability} of noisy PQCs by cost gradient variance analysis. We show that while noise typically induces barren plateaus, an adversary can
balance noise to evade barren plateaus while systematically biasing the optimization landscape.

(4) Through numerical simulations, we validate our theory and show that weakly cross-correlated, adversarial noise accelerates the decay of both the Kraus expressibility norm and gradient variance. In particular, using the hardware-efficient ansatz as a testbed, we demonstrate that increasing noise induces noise-induced barren plateaus even in otherwise trainable shallow circuits, and completely overpowers mitigation strategies such as restricted parameter initialization and local cost functions, ultimately rendering the ansatz untrainable at moderate system sizes.

\emph{\textbf{Related Works}:}
Distributing quantum algorithms~\cite{liu2024nonlocal,beals2013efficient,cirac1999distributed,sadhu2025cutqas} among multiple quantum processors offers a promising pathway to overcome the constraints of NISQ hardware~\cite{ecker2023advances,kim2023evidence,bharti2022noisy}. Distributed QML frameworks have been explored for processing large-scale data~\cite{marshall2023high, hwang2025distributed}, faster training~\cite{du2021accelerating, chen2021federated}, and privacy-preserving learning~\cite{neumann2022distributed}. In the adversarial scenario, recent works have explored integrity oriented threats~\cite{xia2021defending} as well as privacy leakage attacks~\cite{heredge2025characterizing}. Additional QML-specific attack vectors include black-box model extraction via repeated API queries and surrogate training~\cite{fu2024quantumleak,fu2025quantum,fu2025copyqnn}, grey-box  attacks on transpiled circuits to manipulate inputs and estimate parameters or logic ~\cite{wang2023qumos,kundu2025adversarial}, and white-box attacks involving reverse engineering of variational quantum circuits and parameters from transpiled or pulse level representations~\cite{xu2025security,ghosh2025ai,saeed2019locking,rehman2025opaque}, as well as crosstalk induced side channel leakage in the distributed setting. 

 While recent work~\cite{shao2025diagnosing} has explored expressibility and trainability of PQCs under fixed hardware noise, a critical vulnerability remains  unexplored in the context of distributed VQAs: the impact of environmental and adversarial perturbations on the distribution of entangled states~\cite{das2021universal,lu2022micius,sadhu2023practical,wei2022towards} between the QPUs. Since shared entanglement underpins the implementation of non-local quantum gates between different QPUs~\cite{yimsiriwattana2004generalized}, adversarial perturbations at this layer directly affect the computational behavior of VQAs distributed among computationally bounded network nodes~\cite{leone2025entanglement,sadhu2026cryptographic}. This work aims to address this gap.
\begin{figure}
    \centering
    \includegraphics[scale = 0.21]{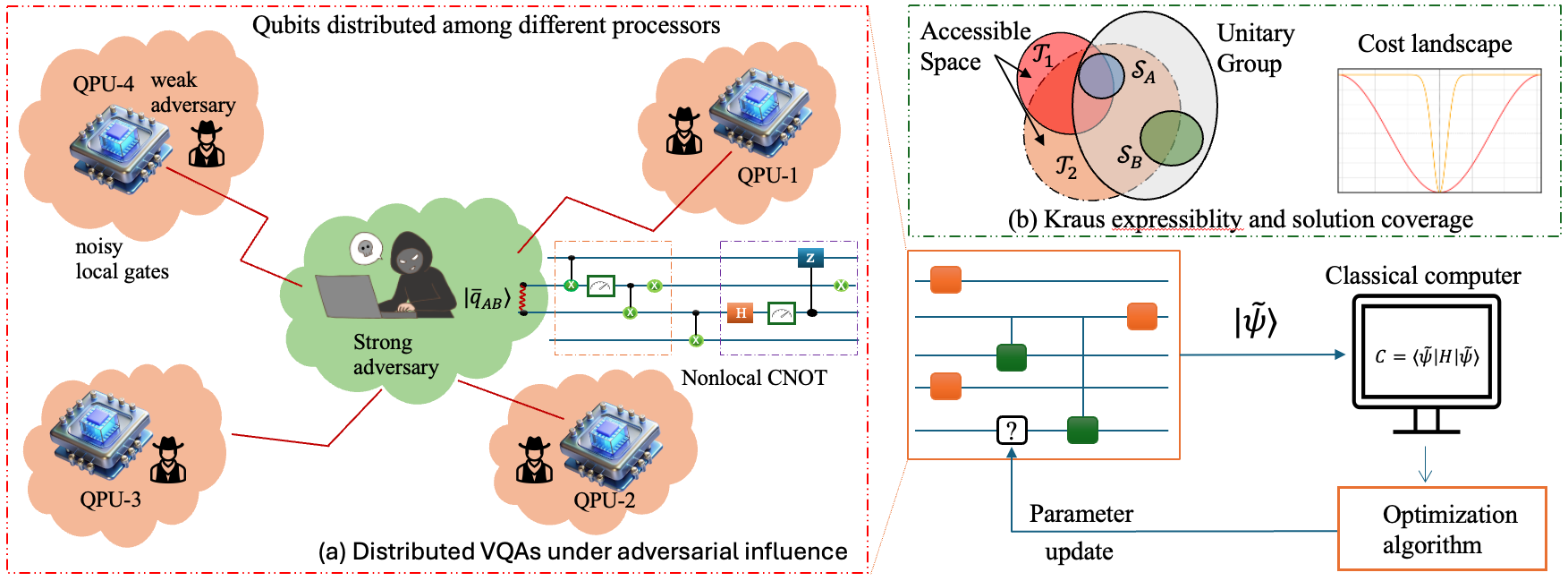}
    \caption{$(a)$ Distributed VQAs under adversarial influence: A strong adversary corrupts shared entanglement among the network nodes, while a weak adversary injects local noise at QPUs, disrupting optimization. $(b)$ Kraus expressiblity and solution coverage: Let $\mathcal{S}_A$ and $\mathcal{S}_B$ denote solution sets within the unitary group for problems A and B. Adversarially-induced noisy non-local gates restrict the accessible transformation space to $\mathcal{T}$. Kraus expressibility measures how uniformly $\mathcal{T}$ explores the unitary group: low expressibility limits coverage (e.g., only $\mathcal{S}_A$), whereas high expressibility enables broader coverage but leads to barren plateaus. }
    \label{fig:vqe_scenario}
    \vspace{-0.3cm}
\end{figure}
\section{Distributed VQAs: Background and Problem Setup}
This section introduces distributed VQAs and formulates the adversarial framework under study.
\subsection{Variational Quantum Algorithms} \label{sec:VQAs}
Variational quantum algorithms (VQAs)~\cite{mcclean2016theory,cerezo2021variational} are hybrid quantum-classical algorithms built on parameterized quantum circuits (PQCs). A PQC defines a parameterized unitary channel $\mathcal{U}(\thetavec) (\cdot)=U(\thetavec)(\cdot)U(\thetavec)^{\dag}$, via the unitary operator $U(\thetavec)$, which acts on an initial $n$-qubit state $\rho_0$ to produce the output state $\rho(\thetavec)=\mathcal{U}(\thetavec)(\rho_0)$. The state is evaluated using a problem-specific Hermitian observable $\mathcal{H}$, yielding the cost $  C_{\rho(\thetavec),\mathcal{H}}={\rm Tr}(\mathcal{H} \rho(\thetavec)), $ where ${\rm Tr}(\cdot)$ denotes the trace operation. The goal of VQA is to find the model parameter $\thetavec \in \Theta$ that minimizes the cost function. This is typically done via classical optimizers using quantum-estimated costs or gradients.

A critical design choice in VQAs is the selection of circuit \textit{ansatz} implementing the channel $\mathcal{U}(\thetavec)(\cdot)$. Typically, this is expressed as a composition of $L$-layers of parameterized unitary channels:
\begin{align}
  \mathcal{U}_{\thetavec}= \mathcal{U}_{L}(\thetavec_L) \circ \mathcal{U}_{L-1}(\thetavec_{L-1}) \hdots \circ \mathcal{U}_{1}(\thetavec_1),  \label{eq:unitarychannels}
\end{align} \vspace{-0.1cm} with $\thetavec=(\thetavec_1,\hdots,\thetavec_{L})$ and $\mathcal{U}_{l} (\thetavec_l)(\cdot)=  U_l(\thetavec_l) W_l (\cdot)W_l^{\dag} U_l(\thetavec_l)^{\dag}$ consisting of fixed unitaries $W_l$ (e.g., CNOT gates), and parameterized rotations $U_l(\thetavec_l)=\exp(-i\thetavec_lV_l)$ generated by Hermitian operators $V_l$. The ansatz defines the hypothesis class explored by the VQA and determines both  expressivity and trainability. 
Common choices include the quantum alternative operator ansatz \cite{farhi2014quantum}, the coupled cluster ansatz \cite{bartlett2007coupled} and the hardware efficient ansatz \cite{kandala2017hardware}.

The VQA formulation above assumes a single $n$-qubit device, limiting scalability. For large-scale problems, the required qubit count often exceeds current hardware capabilities. A natural approach is to distribute computation across multiple quantum processors, enabling the effective simulation of larger systems. We next outline a framework for distributed VQAs.
\subsection{Distributed VQAs via the Cat-Entangler/Cat-Disentangler Protocol} \label{sec:cat-entangler-cat-disentangler}
Distributed VQAs involve partitioning PQCs across multiple QPUs and coordinating their execution. A central challenge in this setting is the implementation of non-local, multi-qubit operations (e.g., CNOT gates) across different QPUs. 
This is typically achieved via distributed protocols that leverage classical communication together with pre-shared entanglement between QPUs \cite{bennett1993teleporting, eisert2000optimal}. In this work, we focus on the \textit{cat-entangler/cat-disentangler} protocol 
\cite{yimsiriwattana2004generalized}, which realizes a non-local CNOT gate with the control and target qubits located on different QPUs, using a single shared entangled Bell pair and two bits of classical communication. We briefly describe this protocol next.
 
Consider a setting where the control qubit $q_c$ resides on QPU A and the target qubit $q_t$ on QPU B. Let $q_A$ and $q_B$ denote the two halves of a shared Bell entangled-state pair,
$\vert q_{AB}\rangle =\frac{1}{\sqrt{2}} (\vert 00 \rangle + \vert 11\rangle)$,  located on QPU A and B, respectively. The protocol proceeds in three steps:\\
\textbf{1. Cat-entangler circuit}:  Applies a local CNOT gate on QPU A with $q_c$ controlling the qubit  $q_A$, 
followed by measurement of $q_A$. If the measurement outcome is $1$, a Pauli-$X$ gate (denoted $\sigma_X$) is applied to $q_B$. This creates a CAT-like entangled state,  $\alpha \vert 00\rangle + \beta \vert 11\rangle$, between $q_c$ and $q_B$, effectively  transferring the control information to $q_B$ on QPU B.
\\
\textbf{2. A local CNOT gate}  is applied on QPU B with $q_B$ as the control and $q_t$ as the target, realizing the desired non-local operation albeit using the entangled qubit $q_B$. 
\\
\textbf{3. Cat-disentangler circuit:} 
Restores the original control qubit by applying a Hadamard gate on qubit $q_B$ followed by measurement. If the  outcome is $1$, a Pauli-$Z$ gate (denoted $\sigma_Z$) is applied to  $q_c$.

Using such non-local gate implementations, distributed VQAs execute PQCs across multiple QPUs in a coordinated manner. Each QPU performs local circuit evaluations and measures its qubits to obtain partial estimates of the cost objective, which are are communicated to a central classical processor.  The processor aggregates them to evaluate the cost function (or its gradients) and update the parameters, which are then broadcast back to the QPUs for the next optimization step. 
\vspace{-0.1cm}
\subsection{Adversarial Attacks on Distributed VQAs} \label{sec:adversarial_attack_levels}
Distributed VQAs based on the cat-entangler/cat-disentangler protocol introduce new attack surfaces at the level of inter-device entanglement. In practical networked settings,  entanglement between spatially separated QPUs is established via repeater-based quantum communication networks, which distribute Bell pairs to enable \textit{non-local} gate operations. Such networks are highly susceptible to both environmental noise as well as adversarial 
attacks \cite{tunc2025resilience, harkness2025security}.

We consider an adversary with partial or total control over the repeater network, capable of tampering with non-local CNOT implementations. Depending on the level at which the perturbation is induced, we distinguish between two adversarial models (see Fig.~\ref{fig:vqe_scenario}$(a)$):
\\
{1. A \textit{weak adversary} operates \textit{locally} at the QPUs and perturbs the computation by \textit{injecting noise} on the control and target qubits before and after the ideal non-local CNOT gate. 
This results in an effective noisy implementation of the gate, without modifying the shared entangled resource.}
\\
{2:}  A \textit{strong adversary} directly manipulates the entanglement distribution process by perturbing the shared Bell pair $\vert q_{AB} \rangle$ (defined in Section~\ref{sec:cat-entangler-cat-disentangler}) as \begin{align}
    \mathcal{A}: \vert q_{AB}\rangle  \rightarrow \vert \tilde{q}_{AB}\rangle= \sum_{i,j \in \{0,1\}}c_{ij}\vert ij\rangle, \label{eq:adversary_attack}
\end{align}  where $\{c_{ij}\}$ are complex coefficients that satisfy $\sum_{i,j} |c_{ij}|^2=1$, excluding the ideal Bell-states.
This perturbation induces deviations in the effective non-local gate, and as shown in Section ~\ref{sec:quantumchannelequivalence}, yields an equivalent noisy quantum channel. 

\begin{wrapfigure}{r}{0.57\textwidth}
    \centering
\includegraphics[scale=0.28, clip=true, trim=0.1in 0.1in 0in 0.1in]{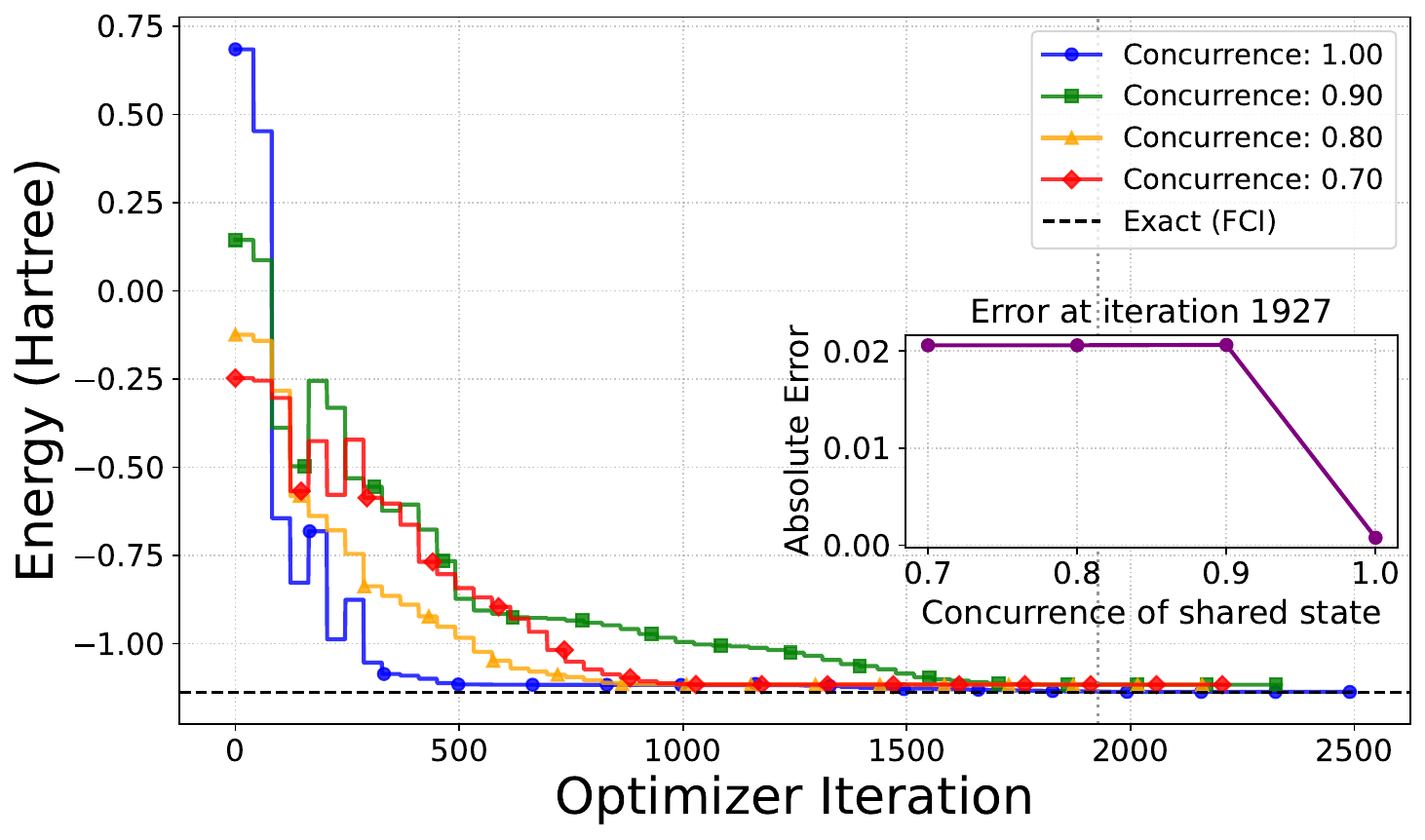}
    \caption{Impact of adversarial perturbations on VQE convergence for the \texttt{H}$_{2}$ ground state.
    }
    \label{fig:vqe}
    \vspace{-0.6cm}
\end{wrapfigure}

The deviations induced by both adversarial models can be unified at the channel level. Let $\mathcal{U} (\cdot)=U (\cdot) U^{\dag}$ denote the ideal non-local CNOT channel, and $\mathcal{K}(\cdot)=\sum_j E_j (\cdot) E_j^{\dag}$ denote the adversarially-induced noise channel, where $\{E_j\}$ are Kraus operators satisfying $\sum_j E_j^{\dag}E_j=\mathbb{I}$.
The adversary aims to degrade the correctness of distributed VQA solutions by inducing a deviation from the ideal unitary channel $\mathcal{U}$, while remaining difficult to detect. 
We quantify detectability via the maximum single-shot discrimination probability,
$
        p_{\text{guess}}(\mathcal{U},\mathcal{K}) \coloneqq \frac{1}{2} (1 + \frac{1}{2} || \mathcal{U} - \mathcal{K} ||_\diamond ), $ where $\Vert \cdot \Vert_{\diamond}$ denotes the diamond norm between quantum channels, and impose the stealth constraint on the adversary: $p_{\text{guess}}(\mathcal{U},\mathcal{K}) \leq \epsilon$ for some $\epsilon \in (0.5,1)$. This constraint captures the requirement that the induced perturbation should not be easily distinguishable from the ideal operation, given the detection capabilities of the network. 
        To illustrate the impact of adversarial perturbations, Fig.~\ref{fig:vqe} considers a distributed variational quantum eigensolver (VQE) for computing the ground state of 4-qubit \texttt{H}$_2$ molecule in equilibrium geometry considering sto-3g basis set and Jordan-Wigner mapping. The objective is to minimize the cost $  C_{\ket{\psi(\thetavec)},\mathcal{H}_{\texttt{H}_2}}=\bra{\psi(\vec{\theta})} \mathcal{H}_{\texttt{H}_2} \ket{\psi(\vec{\theta})}$, where the state $\ket{\psi(\vec{\theta})}$ is prepared by applying a hardware efficient ansatz to the initial state $\ket{\psi_{0}}$.
We model a strong adversary that perturbs non-local CNOT gates as in Proposition~\ref{prop:noisy_cnot_cat}. 
Fig.~\ref{fig:vqe} shows  the resulting energy landscape across the  optimization iterations.
In the adversary-free setting (concurrence\footnote{We quantify the amount of entanglement in a bipartite quantum state $\rho_{AB}$ using concurrence~\cite{wootters1998entanglement} which is defined as $\max\{0,\lambda_1^e - \lambda_2^e - \lambda_3^e - \lambda_4^e\}$ where $\lambda_s^e$ denotes the square root of the eigenvalues of $\rho_{AB} \tilde{\rho}_{AB}$ with $\tilde{\rho}_{AB} = (\sigma_Y \otimes \sigma_Y) \rho_{AB} (\sigma_Y \otimes \sigma_Y)$ being the spin-flipped state of $\rho_{AB}$ and $\sigma_Y$ being the Pauli-Y matrix.} $=1$), the optimizer converges to the correct ground-state energy. In contrast, as the adversary reduces concurrence (i.e., increases noise), optimization converges to a biased solution and fails to reach chemical accuracy.

While such deviations could, in principle, be detected via verification or spot-checking in fault-tolerant settings, this is impractical in resource-constrained NISQ scenarios, where each VQE iteration is costly. Consequently, adversarial perturbations can silently degrade optimization performance.
\vspace{-0.1cm}
\section{Adversarial Attacks and Quantum Channel Equivalence}\label{sec:quantumchannelequivalence} 

We now characterize the effect of \textit{strong adversarial} perturbations on non-local CNOT implementations. In particular, we show that perturbations to the shared entangled state in the cat-entangler/cat-disentangler protocol induce an equivalent noisy quantum channel acting on the local qubits.
\begin{proposition} \label{prop:noisy_cnot_cat}
    Consider the cat-entangler/cat-disentangler protocol for implementing a non-local CNOT gate between control qubit $q_c$ (QPU A) and target qubit $q_t$ (QPU B). Suppose a strong adversary perturbs the Bell-entangled state pair as in Eq.~\eqref{eq:adversary_attack}. Then, the resulting implementation of the non-local CNOT operation is equivalent to a noisy quantum channel $\mathcal{K}^{\mathbf{p}}(\cdot)= \sum_{i,j \in \{0,1\}} E_{ij}(\cdot)E_{ij}^{\dag}$ with parameters $\mathbf{p}=(c_{00},c_{01},c_{10},c_{11})$ and the following Kraus operators:
    \begin{eqnarray}
    && E_{00} = \frac{1}{\sqrt{2}} \op{0}_{q_c} \otimes (c_{00} \mathbbm{1} + c_{01} \sigma_X)_{q_t} + \frac{1}{\sqrt{2}} \op{1}_{q_c} \otimes (c_{10} \mathbbm{1} + c_{11} \sigma_X)_{q_t}, \\
    && E_{01} = \frac{1}{\sqrt{2}} \op{0}_{q_c} \otimes (c_{00} \mathbbm{1} - c_{01} \sigma_X)_{q_t} - \frac{1}{\sqrt{2}} \op{1}_{q_c} \otimes (c_{10} \mathbbm{1} - c_{11} \sigma_X)_{q_t}, \\
    && E_{10} = \frac{1}{\sqrt{2}} \op{0}_{q_c} \otimes (c_{11} \mathbbm{1} + { c_{10}} \sigma_X)_{q_t} + \frac{1}{\sqrt{2}} \op{1}_{q_c} \otimes (c_{01} \mathbbm{1} + c_{00} \sigma_X)_{q_t}, \\
    && E_{11} = \frac{1}{\sqrt{2}} \op{0}_{q_c} \otimes (c_{11} \mathbbm{1} - { c_{10}} \sigma_X)_{q_t} - \frac{1}{\sqrt{2}} \op{1}_{q_c} \otimes (c_{01} \mathbbm{1} - c_{00} \sigma_X)_{q_t}.
    \end{eqnarray}
\end{proposition}
 Proposition 1 establishes that adversarial perturbations at the level of shared entanglement translate directly into structured noise at the level of non-local gate operations. In fact, the resulting noisy channel $\mathcal{K}^{\mathbf{p}}(\cdot)$ depends explicitly on the perturbation  parameters $\mathbf{p}$ providing a principled mapping from entanglement-level attacks to gate-level noise. Detailed proof can be found in Appendix~\ref{app:noisy_cnot_cat}.
 
 We now consider two special cases. In the ideal adversary-free setting, where $c_{00} = 1/\sqrt{2} = c_{11}$ and $c_{01} = 0 = c_{10}$,  all Kraus operators coincide with $E_{ij}=\frac{1}{2}U_{\rm CNOT}$, for all $i,j$, where $U_{\rm CNOT}= \op{0}_{q_c} \otimes \mathbbm{1}_{q_t} +  \op{1}_{q_c} \otimes (\sigma_X)_{q_t}$. Hence, $\mathcal{K}^{\pbf}$ reduces to the ideal unitary CNOT channel $\mathcal{K}^{\pbf}(\cdot)=U_{\rm CNOT}(\cdot) U_{\rm CNOT}^{\dag}.$  In contrast, when $c_{01} = 1/\sqrt{2} = c_{10}$ and $c_{00} = 0 = c_{11}$, we get $E_{00}=E_{10}=\frac{1}{2}U_{\rm flip}$ and $E_{01}=E_{11}=-\frac{1}{2}U_{\rm flip}$, where  $U_{\rm flip}= \op{0}_{q_c} \otimes (\sigma_X)_{q_t} +  \op{1}_{q_c} \otimes \mathbbm{1}_{q_t}$. Thus, the channel $\mathcal{K}^{\pbf}$ reduces to the flipped CNOT unitary channel up to a global phase. 

\section{Kraus Expressibility  for Noisy Quantum Circuits}
As discussed in Section~\ref{sec:VQAs}, the success of  VQAs depends critically on the choice of the sequence of transformations applied to the initial state. In standard VQAs, these are parameterized unitary channels (see Eq.~\eqref{eq:unitarychannels}). However, in the presence of adversarially-induced noise, the effective transformations are no longer unitary, but quantum channels. To capture this, we generalize the standard unitary-ansatz to a parameterized quantum channel \textit{(PQCh)-ansatz}.

We define a PQCh-ansatz {of depth $L$} as a composition of $L$ quantum channels: \begin{align}\mathcal{E}_{\vec{\theta}}^{\mathbf{p}} = \Kcal_L (\mathbf{p}_L,\vec{\theta}_L) \circ \Kcal_{L-1} (\mathbf{p}_{L-1},\vec{\theta}_{L-1}) \circ ... \circ \Kcal_1 (\mathbf{p}_1,\vec{\theta}_1), \label{eq:PQCh-ansatz} \end{align} where $\mathbf{p}_j$ and $\vec{\theta}_j$ denote the noise and trainable parameters of the $j$-th channel $\Kcal_j$, respectively.
When all $\Kcal_j$'s collapse to noise-free unitary channels, the PQCh-ansatz reduces to the standard ansatz in  Eq.~\eqref{eq:unitarychannels}. For fixed noise parameter $\mathbf{p}=(\mathbf{p_1},\hdots,\mathbf{p}_L)$, varying $\vec{\theta}$ induces an ensemble $\mathcal{T}_{\pbf}=\{\mathcal{E}^{\mathbf{p}}_{\vec{\theta}}: \thetavec \in \Theta \}$ of fixed-noise trainable channels. Further variation over noise parameters $\pbf \in \Pi$ yields the full ensemble $\mathcal{T}=\{\mathcal{T}_{\pbf}: \pbf \in \Pi\}$ of noisy trainable transformations.

VQA success requires that the target solution to a problem lies within the ansatz-induced ensemble $\mathcal{T}$, i.e., $\mathcal{S} \cap \mathcal{T} \neq \emptyset$, where $\mathcal{S}$ denotes the set of solution transformations. In practice, $\mathcal{S}$ is unknown, making this condition difficult to verify. A common design principle, widely adopted in problem-agnostic settings, is therefore to ensure that the ansatz explores the transformation space broadly and uniformly, increasing the likelihood of approximating the optimal solution. This property is known as \textit{expressibility} \cite{sim2019expressibility, nakaji2021expressibility, holmes2022connecting}. 

 In the adversarial setting,  the adversarially-induced noise can distort this exploration and drive the PQCh-ansatz away from the optimal transformation. Consequently, even if the noiseless ansatz can cover the solution transformation, the adversarially-perturbed PQCh-ansatz may not. This motivates a notion of expressibility that explicitly accounts for noise and adversarial effects. We formalize this via \textit{Kraus expressibility}, which quantifies how well noisy quantum circuits explore the space of ideal unitary transformations (see Fig.~\ref{fig:vqe_scenario}\textcolor{blue}{$(\text{b})$} for an illustration).

\begin{definition}
    The \textit{Kraus expressibility} of the PQCh-ansatz is defined via the following super-operator,
    \begin{eqnarray}
    \mathcal{A}_{\mathcal{T}}^t (\cdot) \coloneqq \int_{\mathcal{U}(d)} d\mu (V) V^{\otimes t} (\cdot) (V^\dag)^{\otimes t} 
    - \int_{\mathcal{T}} d\nu(E) \sum_{k_1...k_l} (\otimes_{j=1}^t E_{k}^{(j)}) ~ (\cdot) ~ (\otimes_{j=1}^t E_{k}^{(j)\dag}), \label{eq:expressivity_operator_main_text}
\end{eqnarray}
where $d\mu(V)$ is the volume element of the Haar measure over the unitary group $\mathcal{U}(d)$ with $d=2^n$, $d\nu(E)$ is the volume element corresponding to uniform distribution over the PQCh-ensemble $\mathcal{T}$, and $E_{k}^{(j)}$ denote the k-th Kraus operator of the PQCh-ansatz for the $j$-th copy.
\end{definition}
Intuitively, Kraus expressibility quantifies how well the PQCh-ensemble $\mathcal{T}$ uniformly explores the unitary group $\mathcal{U}(d)$.  In particular, if $\mathcal{A}_{\mathcal{T}}^t (X)=0$ for all operators $X$, then averaging over $\mathcal{T}$ matches Haar averaging  over $\mathcal{U}(d)$ upto the $t$-th moment, i.e., $\mathcal{T}$ forms a $t$-design. In this work, we focus on $t=2$. 
For VQA optimization, we evaluate the Kraus expressibility of $\mathcal{T}$ with respect to the input state $\rho$ via the norm
\begin{eqnarray}
\Delta_{\mathcal{T}}^\rho \coloneqq || \mathcal{A}_{\mathcal{T}}^2 (\rho^{\otimes 2}) ||_2.
\end{eqnarray}
Smaller values of $\Delta_{\mathcal{T}}^\rho$ indicate higher Kraus expressibility. This extends unitary expressibility \cite{holmes2022connecting} to noisy and  adversarial settings.
\begin{theorem} \label{theorem:kraus_expressibility}
The  Kraus expressibility norm of the PQCh-ansatz with respect to the state $\rho$ is given by   
    \begin{eqnarray}
        (\Delta_{\mathcal{T}}^\rho)^2 = (\alpha^2 + \beta^2) ~ d^2 + 2 \alpha \beta d - 2 [\alpha + \beta \bar{\nu}] + \mathcal{N}_{\text{noise}}, \label{eq:Kraus_norm}
    \end{eqnarray} 
    where $\alpha = \frac{d - \Tr \rho^2}{d(d^2 - 1)}$, $\beta = \frac{d \Tr \rho^2 - 1}{d(d^2 - 1)}$, $\bar{\nu}$ is the ensemble averaged purity \footnote{A quantum state $\rho$ is pure if ${\rm Tr}(\rho^2)=1$ and is mixed if ${\rm Tr}(\rho^2)<1$.} of the input state $\rho$ after evolving through the PQCh-ansatz, and $\mathcal{N}_{\text{noise}}$ is the noise fluctuation term defined as
    $\mathcal{N}_{\text{noise}} = \int_{\mathcal{T}}\int_{\mathcal{T}}d\tau(E)d\nu(F)\left[\sum_{j,k} \Tr(E_{k}\rho E_{k}^{\dagger}F_{j}\rho F_{j}^{\dagger})\right]^2$
    with $\{E_k\}$ and $\{F_k\}$ representing Kraus operators corresponding to independent channel realizations sampled from 
    $\mathcal{T}$. 
\end{theorem}
Proof is provided in  Appendix~\ref{app:Kraus_expressivity}.  Theorem~\ref{theorem:kraus_expressibility} shows that adversarial effects  influence the Kraus expressibility through two distinct mechanisms:
$(i)$ Average state purity loss: As the state $\rho$ evolving through the PQCh-ansatz becomes increasingly mixed, average state purity $\bar{\nu}$ decreases, reducing the Kraus expressibility, and  $(ii)$ Diversity loss due to noise-induced correlations: The term $\mathcal{N}^{\rm noise}$ captures the  expected squared overlap between states evolved under independent channel realizations sampled from $\mathcal{T}$. Hence, high correlations across channel realizations in the ensemble (e.g., under strongly mixing noise such as depolarizing noise) reduces the diversity of the explored transformations, thereby lowering Kraus expressibility. In particular, in the high depolarizing noise regime (noise strength $\rightarrow 1$), outputs concentrate toward the maximally mixed state, leading to reduced purity and increased overlap across realizations. Both effects increase $(\Delta_{\mathcal{T}}^\rho)^2$, thereby reducing expressibility.

Intuitively, a highly Kraus expressive PQCh-ensemble  requires both high average purity (to preserve the quantum structure) and sufficiently diverse channel realizations (to explore the transformation space). Adversarial noise can disrupt either (or both), making the resulting ansatz less capable of approximating target solutions. Notably, even with negligible purity loss, strong noise-induced correlations can significantly degrade expressibility. In the absence of adversarial noise, these effects vanish, and Kraus expressibility reduces to standard unitary expressibility.

While Theorem~\ref{theorem:kraus_expressibility} characterizes the influence of adversarial perturbations over the full space of noisy trainable transformations, practical VQA deployments often operate in a post-training (inference) regime, where circuit parameters are optimized offline and subsequently fixed. In this setting, the adversary acts \textit{after training}, 
and thus perturbs the induced channel via  the noise parameter $\pbf$. The following corollary characterizes Kraus expressibility in this regime.
\begin{corollary} \label{corollary:expressibility_limiting_case}
    When the PQCh-ensemble reduces to a composite channel $\mathcal{E}_{\thetavec}^{\pbf}$ with fixed  $(\pbf,\thetavec)$, the  Kraus expressibility norm with respect to the input state $\rho$ is given by
    \begin{eqnarray}
        (\Delta_{\mathcal{E}_{\thetavec}^{\pbf}}^\rho)^2 = (\alpha^2 + \beta^2) d^2 + 2 \alpha \beta d - 2 [\alpha + \beta \nu] + \nu^2, \label{eq:posttraining_Krausexpressibility}
        \vspace{-0.1cm}
    \end{eqnarray}
    where $\nu$ is the purity of the state after evolving through $\mathcal{E}_{\thetavec}^{\pbf}$, and $(\alpha,\beta)$ are as defined in Theorem~\ref{theorem:kraus_expressibility}.
\end{corollary}

It follows that, in the post-training regime, the expressibility norm $ (\Delta_{\mathcal{E}_{\thetavec}^{\pbf}}^\rho)^2$ depends solely on the output state purity, which is governed by  the noise parameters $\pbf$. In this regime, the ensemble-level noise correlation term $\mathcal{N}_{\rm noise}$ in Theorem~\ref{theorem:kraus_expressibility} collapses to a function of $\nu$. Moreover, since $(\Delta_{\mathcal{E}_{\thetavec}^{\pbf}}^\rho)^2$ is monotonically increasing in $\nu$,  increasing noise (which reduces purity) decreases $(\Delta_{\mathcal{E}_{\thetavec}^{\pbf}}^\rho)^2$. 
\vspace{-0.3cm}
\section{Trainability-Expressibility Trade-offs for Noisy Quantum Circuits}

While Kraus expressibility determines whether a PQCh-ansatz can represent the target solution, it does not guarantee that the solution can be \textit{efficiently learned}. Successful optimization additionally requires the ansatz to be \textit{trainable}, i.e., the cost landscape must exhibit sufficiently large gradients for the solution to be found. 

For a PQCh-ansatz as in Eq.~\eqref{eq:PQCh-ansatz} with input state $\rho_0$, we define the cost function as \begin{align}C_{\rho_0,\mathcal{H}}={\rm Tr}(\mathcal{H}~\mathcal{E}^{\pbf}_{\thetavec}(\rho_0)). \label{eq:cost}\vspace{-0.1cm}\end{align}  Then, the gradient of the cost function with respect to the $k$-th trainable parameter $\thetavec_k$ is denoted by the partial derivative $\partial_k C_{\rho_0,\mathcal{H}}:= \partial C_{\rho_0,\mathcal{H}}(\thetavec)/\partial{\thetavec_k}$. Assuming $\thetavec_k \sim {\rm Unif} [0,2\pi]$, it can be shown that
$ \Ebb_{\thetavec_k}[\partial_k C_{\rho_0,\mathcal{H}}]=0$ for all $k$  (see Appendix~\ref{app:quantum_channel_expressivity}). In this unbiased cost landscape, the trainability of the ansatz is governed by the \textit{gradient variance}: ${\rm var}_{\thetavec_k}(\partial_k C_{\rho_0,\mathcal{H}} ) =\Ebb_{\thetavec_k}[(\partial_k C_{\rho_0,\mathcal{H}})^2 ], $ which quantifies the extent to which the gradient fluctuates away from zero. In particular, if the variance decays exponentially with the number of qubits $n$, the cost landscape exhibits \textit{barren plateaus}, rendering training intractable \cite{mcclean2018barren}.

 We will next relate the Kraus expressibility to the gradient variance. To formalize this, we decompose the PQCh-ansatz in Eq.~\eqref{eq:PQCh-ansatz} around the $k$-th parameter as 
$\mathcal{E}^{\pbf}_{\thetavec}= \mathcal{E}_{\Lcal} \circ \mathcal{K}_k(\pbf_{k},\thetavec_{k}) \circ \mathcal{E}_{\Rcal},  $
where $\mathcal{E}_{\Lcal}= \mathcal{K}_{L}(\pbf_L,\thetavec_L) \circ \hdots \mathcal{K}_{k+1}(\pbf_{k+1},\thetavec_{k+1})$ and 
$\mathcal{E}_{\Rcal}= \mathcal{K}_{k-1}(\pbf_{k-1},\thetavec_{k-1}) \circ \hdots \mathcal{K}_{1}(\pbf_{1},\thetavec_{1})$ denote the left and right channel compositions, respectively. Let $\mathcal{T}_{\Lcal}$ and $\mathcal{T}_{\Rcal}$ denote the corresponding channel ensembles. 
We further define the reference gradient variance ${\rm var}_{\mathcal{R}}(\partial_k C_{\rho_0,\mathcal{H}}):=\Ebb_{\mathcal{E}_{\Lcal}\sim \mathcal{T}_{\Lcal}, \mathcal{E}_{\Rcal} \sim \mathcal{U}_{\Rcal}}[{\rm var}_{\thetavec_k}(\partial_k C_{\rho_0,\mathcal{H}}) ]$ \footnote{We write $\mathcal{E} \sim \mathcal{T}$ to denote that the channel $\mathcal{E}$ is sampled from the ensemble $\mathcal{T}$.} obtained when the right ensemble $\mathcal{T}_{\Rcal}$ forms a unitary 2-design $\mathcal{U}_{\mathcal{R}}$. This corresponds to the setting where the right subcircuit is maximally expressive (Haar random), and yields variance scaling  as $\mathcal{O}(2^{-2n})$, decaying exponentially with $n$ \cite{mcclean2018barren, holmes2022connecting}.

The following theorem bounds the deviation of the expected gradient variance from this reference via the Kraus expressibility norm. 
\begin{theorem} \label{theorem:trainability_expressibility}
    Consider the generic cost function $C_{\rho_0,\mathcal{H}}$ defined in Eq.~\eqref{eq:cost} for a PQCh-ansatz as in Eq.~\eqref{eq:PQCh-ansatz}. Then,  the deviation of the expected gradient variance from the reference variance $\mathrm{var}_{\Rcal}(\partial_k C_{\rho_0, \mathcal{H}})$ is upper bounded as
    \begin{eqnarray}
    |\mathbb{E}_{\mathcal{E}_{\Lcal} \sim \mathcal{T}_{\Lcal}, \mathcal{E}_{\Rcal} \sim \mathcal{T}_{\Rcal}}[{\rm var}_{\thetavec_k}(\partial_k C_{\rho_0, \mathcal{H}})] - {\rm var}_{\Rcal}(\partial_k C_{\rho_0, \mathcal{H}})|  
    \le 4 \Vert \mathcal{A}^2_{\mathcal{T_R}}(\rho_0^{\otimes 2})\Vert_2 \int_{\mathcal{T}_{\Lcal}} d\nu(\mathcal{E}_{\Lcal}) ||\mathcal{E}_{\Lcal}^\dagger(\mathcal{H})||_2^2 . \label{eq:trainability_expressivity_theorem}
\end{eqnarray}
\end{theorem} 
Proof is provided in  Appendix~\ref{app:noisy_gradient_evolution}. Theorem~\ref{theorem:trainability_expressibility} provides a direct, analytical link between the trainability and expressibility of the PQCh-ansatz. 
 The deviation of the expected gradient variance from the reference ${\rm var}_{\Rcal}(\partial_k C_{\rho_0,\mathcal{H}}) $ is governed by two factors : $(i)$ the Kraus expressibility norm of the right sub-circuit evaluated on the input state, and $(ii)$ the observable attenuation $||\mathcal{E}_{\Lcal}^\dagger(\mathcal{H})||_2^2$, capturing the decay of the measurement signal when propagated backward through the left noisy channel. In the highly Kraus expressive regime for the preceding layers, i.e., $||\mathcal{A}_{\mathcal{T_R}}^2(\rho_0^{\otimes 2})||_2 \rightarrow 0$, the deviation vanishes and the expected gradient variance concentrates around the exponentially small baseline $\mathcal{O}(2^{-2n})$ of ${\rm var}_{\Rcal}(\partial_k C_{\rho_0,\mathcal{H}}) $, recovering the barren plateau behaviour. This characterization also explains the emergence of \textit{noise-induced barren plateaus} \cite{wang2021noise}. Noise, whether intrinsic or adversarially-induced, can 
 attenuate the observable through the factor $\Vert \mathcal{E}_{\Lcal}^{\dag}(\mathcal{H})\Vert_2$, leading to vanishing gradient variance even when the Kraus expressibility norm is non-negligible.  This occurs under strong depolarizing noise, where the Kraus expressibility norm remains large (low expressibility), while the attenuation factor for traceless observables decays exponentially with circuit depth.

We summarize the impact of strong adversarial attacks on distributed VQAs. A strong adversary perturbs shared entanglement to induce noise in non-local gates, reducing $\bar{\nu}$ by $\delta$ while remaining $\epsilon$-indistinguishable from ideal CNOT implementation. For a PQCh-ansatz with $L$ layers, this results in $\bar{\nu} \approx (1 - \mathcal{O}(\delta^2))^L$. 
While reduced purity affects Kraus expressibility, Theorem~\ref{theorem:trainability_expressibility} reveals a layered asymmetry in trainability.
For  parameters near the end of the PQCh-ansatz, trainability is bottlenecked by preceding layers approaching a 2-design, driving the expressibility norm $||\mathcal{A}_{\mathcal{T}_R}^{(2)} (\rho_0^{\otimes 2})||_2$ towards zero. In contrast, for parameters near the beginning of the ansatz, the trainability is limited by the observable attenuation factor, $\|\mathcal{E}_L^\dagger(\mathcal{H})\|_2^2$. A budget-constrained adversary can exploit this symmetry: attacking early layers flatten the cost landscape via signal attenuation, while attacking later layers suppress gradients through increased expressibility. By calibrating these effects, the adversary can systematically bias the optimization objective, degrading the final solution quality.
\vspace{-0.2cm}
\section{Numerical simulations}
We now present numerical simulations to study the effect of noise on the Kraus expressibility and the cost gradient scaling. We consider a 
2-local cost function  $\mathcal{H}_L \coloneqq \sigma_Z^1 \otimes \sigma_Z^2$, with $\sigma_Z$ measurement on the first two qubits. The input state is $\rho_0 = \op{\psi_0}^{\otimes n}$,  where $\ket{\psi_0} = e^{-i \pi \sigma_y/8} \ket{0}$. We employ a hardware-efficient PQCh-ansatz comprising of $L$ layers of single qubit rotation gates ($R_y$ and $R_z$ in sequence) followed by a noisy entangling layer. The entangling layer consists of a ladder of noisy CNOT gates, with noise adversarially-induced as in Proposition~\ref{prop:noisy_cnot_cat} and controlled via the concurrence of the shared entangled state. In the noise-free and maximally expressive version of the PQCh-ansatz, rotation angles are sampled uniformly from $[0,2\pi]$. Simulations are performed using Qulacs~\cite{suzuki2021qulacs} and standard Python numerical programming libraries. 

\begin{figure}
    \centering
    \subfigure[]{\label{subfig:expressibility_vs_depth_plot}\includegraphics[width=0.32\textwidth]{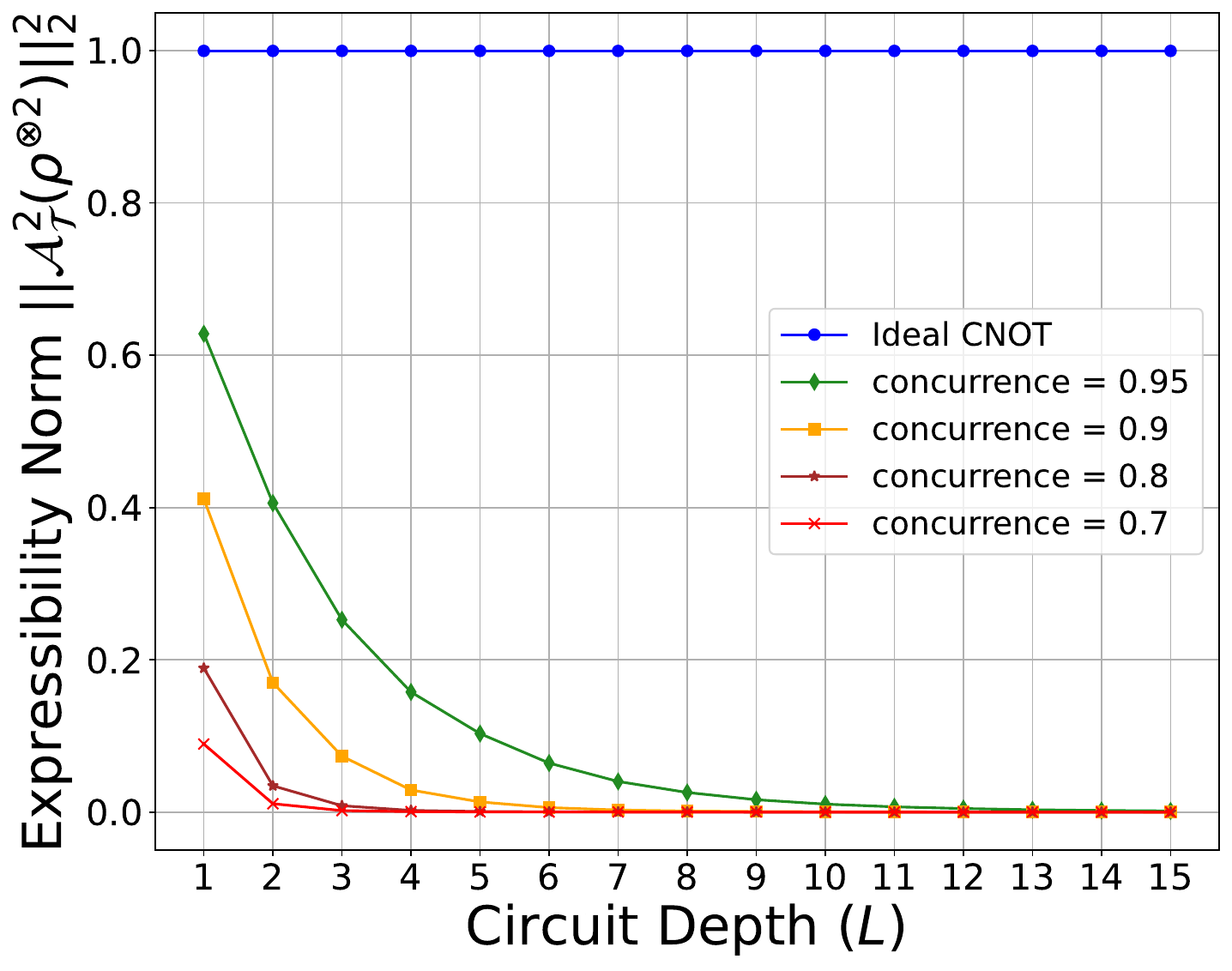}}
    \subfigure[]
    {\label{subfig:variance_vs_depth_plot} \includegraphics[width=0.32\textwidth]{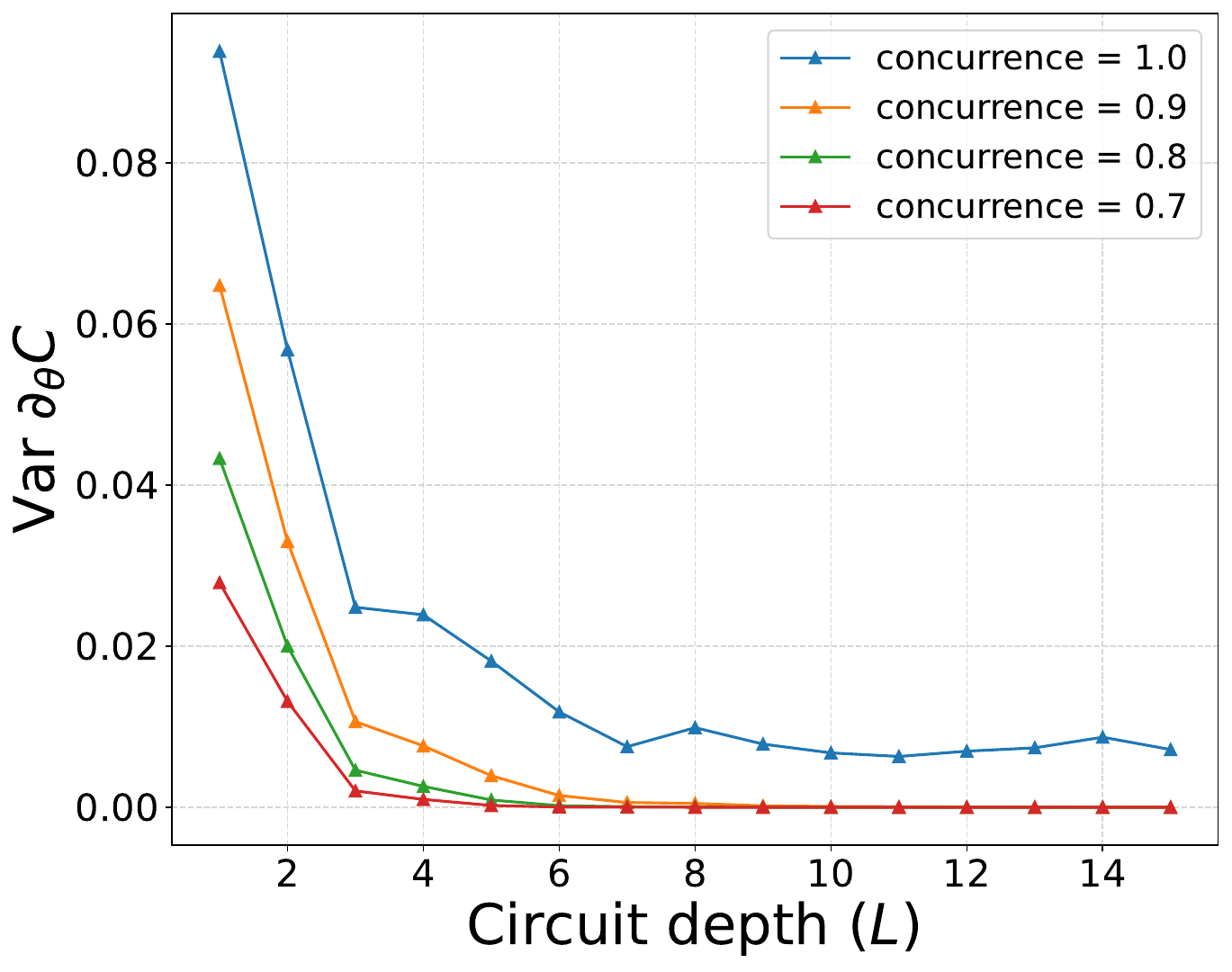}} 
    \subfigure[]{\label{subfig:variance_vs_concurrence_plot}\includegraphics[width=0.32\textwidth]{ 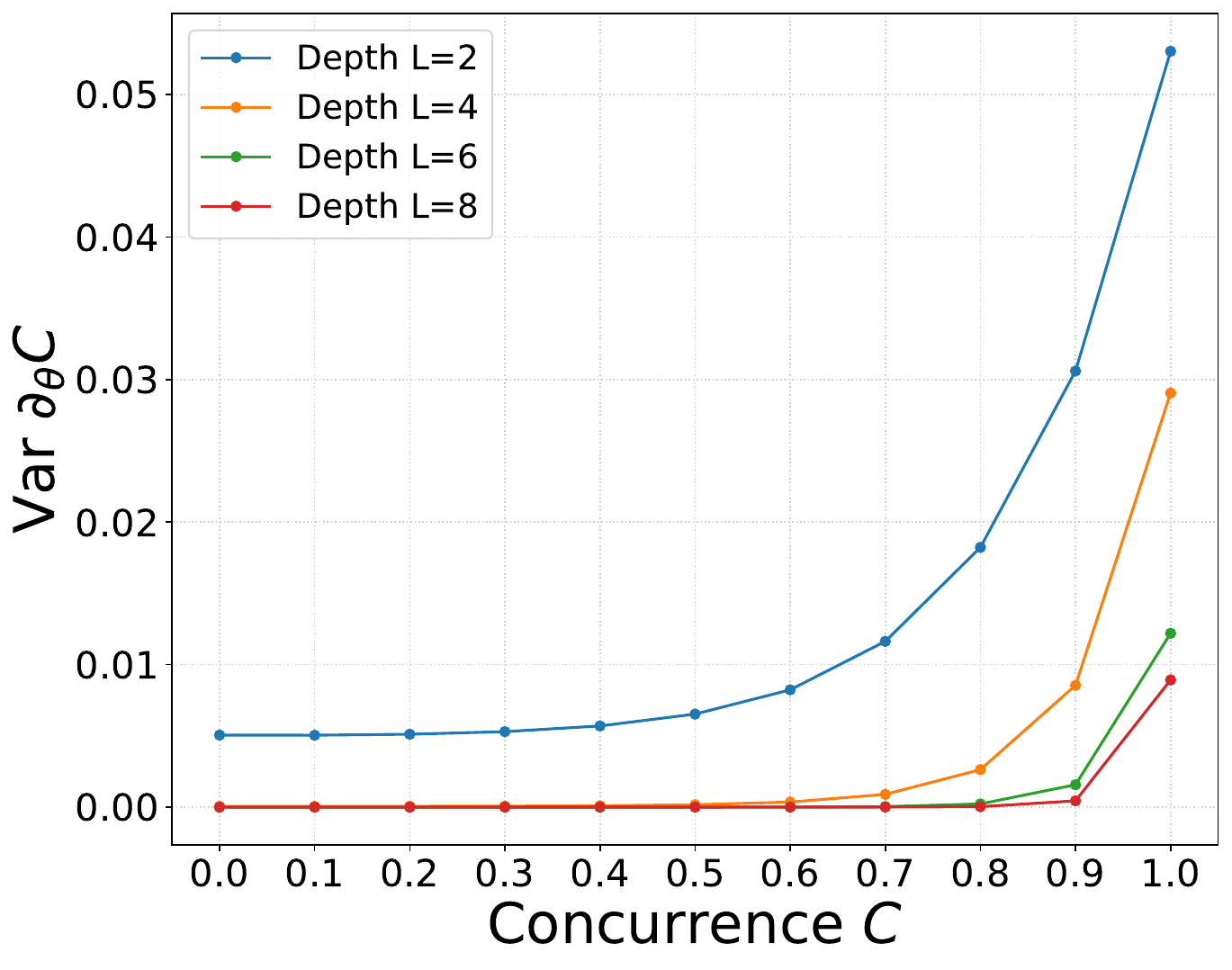}} 
    \caption{Adversarial impact on trainability  for $n=6$ qubit systems. (a) For fixed trainable parameters, Kraus expressibility norm decays with PQCh-ansatz depth and adversarial noise. (b) Adversarial noise in non-local CNOTs (via reduced concurrence) renders even shallow circuits untrainable. (c) At fixed depth, stronger adversarial noise (low concurrence) reduces gradient variance thereby impeding optimization. The plots are averaged over 100 trials.}
    \label{fig:foobar}
    \vspace{-0.6cm}
\end{figure}

\emph{Effect of Depth and Noise on Kraus Expressibility:} 
We first fix the trainable parameters of the PQCh-ansatz and introduce noise in the CNOT layers via the Kraus operators in Proposition~\ref{prop:noisy_cnot_cat}.
Fig.~\ref{subfig:expressibility_vs_depth_plot} shows that the  Kraus expressibility  norm decreases with circuit depth (indicating increasing Kraus expressibility), and this effect is amplified by increasing noise (or lowering concurrence). 
This observation is consistent with Corollary~\ref{corollary:expressibility_limiting_case}: noise reduces output state purity, which in turn lowers Kraus expressibility norm, with compounding effects across layers.

\emph{Gradient Variance and Barren Plateaus:} Prior work \cite{cerezo2021cost} shows that (in the noise-free setting) hardware-efficient ans{\"a}tze (HEA) with local cost function exhibit barren plateau for deep circuits but are trainable for shallow circuits. We observe the same trend in the noise-free scenario in Fig.~\ref{subfig:variance_vs_depth_plot}: the gradient variance rapidly drops to zero with increasing depth, rendering deep circuits untrainable.

Adversarial noise further amplifies this effect.  As the concurrence of the shared entangled state decreases (i.e., noise increases), the gradient variance diminishes even for shallow circuits as shown in Fig.~\ref{subfig:variance_vs_depth_plot}, indicating \textit{noise-induced barren plateaus}. Furthermore, as shown in Fig.~\ref{subfig:variance_vs_concurrence_plot}, for fixed depth, increasing the noise in CNOT gate reduces the gradient variance, making it more difficult to train. This aligns with Theorem~\ref{theorem:trainability_expressibility}: increasing noise impacts both the Kraus expressibility norm and the signal strength via channel attenuation, resulting in vanishing gradients.  

\begin{figure}
    \centering
    \subfigure[]{\label{subfig:variance_cost_vs_num_qubit_ideal}\includegraphics[width=0.32\textwidth]{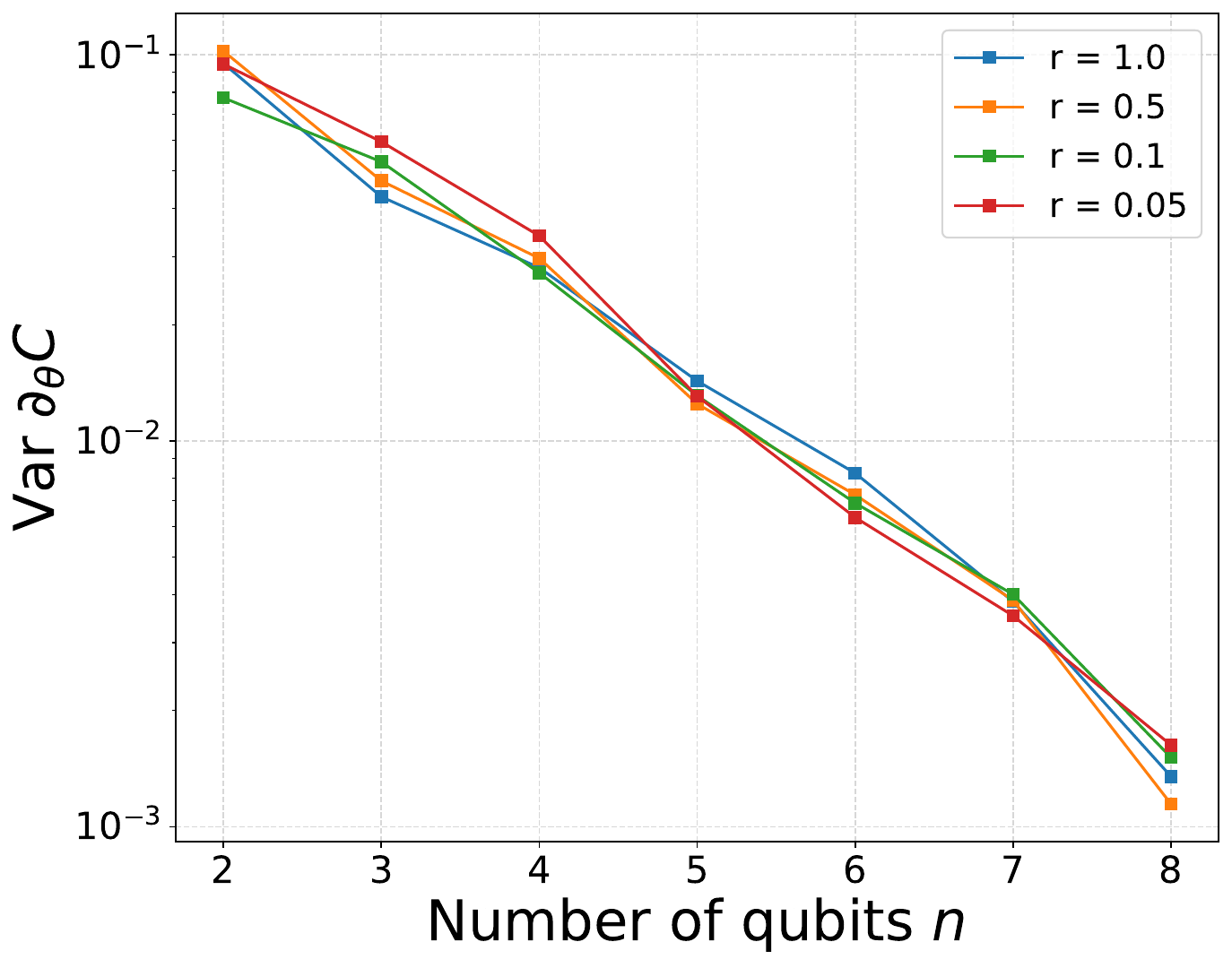}} 
    \subfigure[]
    {\label{subfig:variance_cost_vs_num_qubit_conc_point_nine}\includegraphics[width=0.32\textwidth]{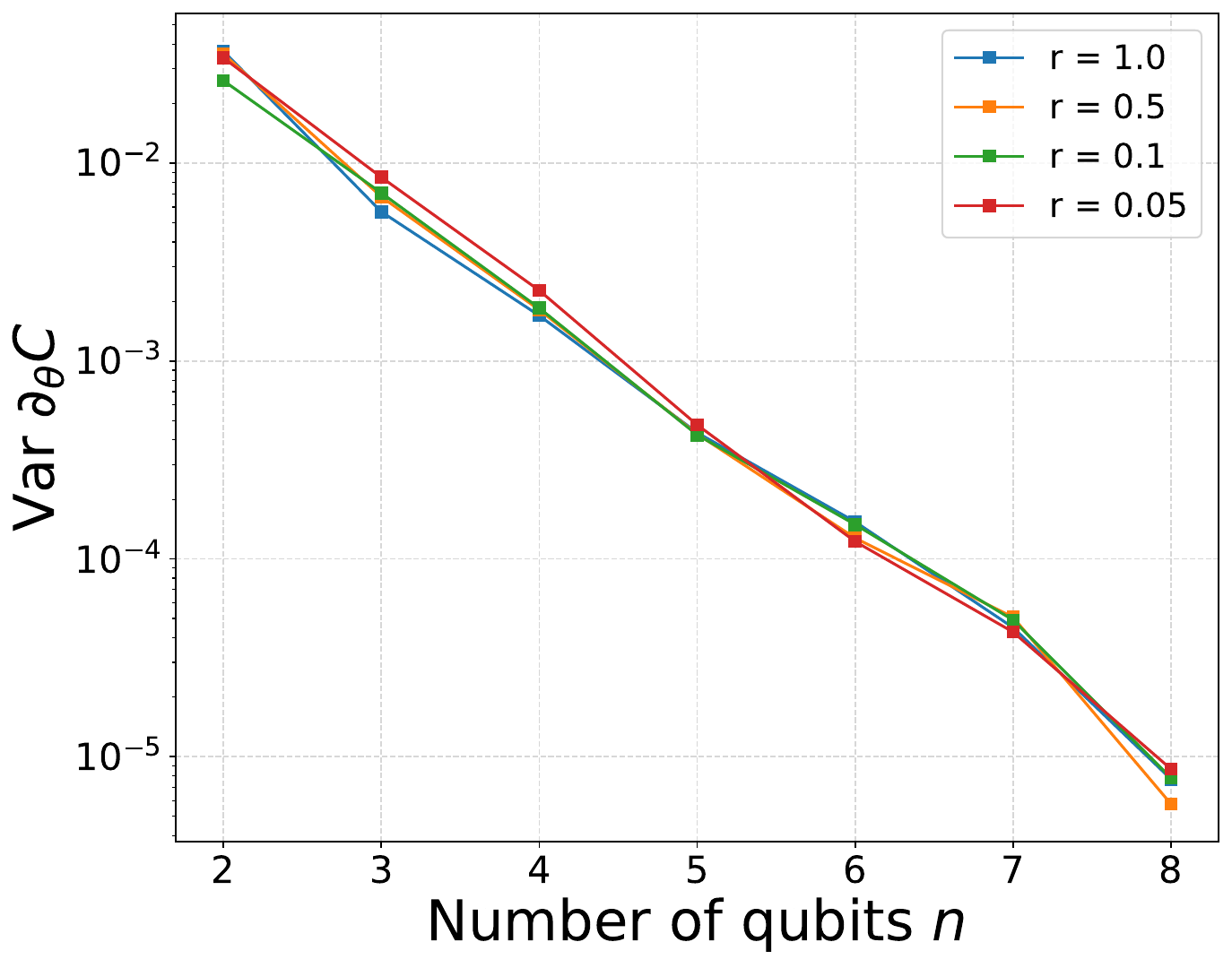}} 
    \subfigure[]
    {\label{subfig:variance_cost_vs_num_qubit_conc_point_eight}\includegraphics[width=0.32\textwidth]{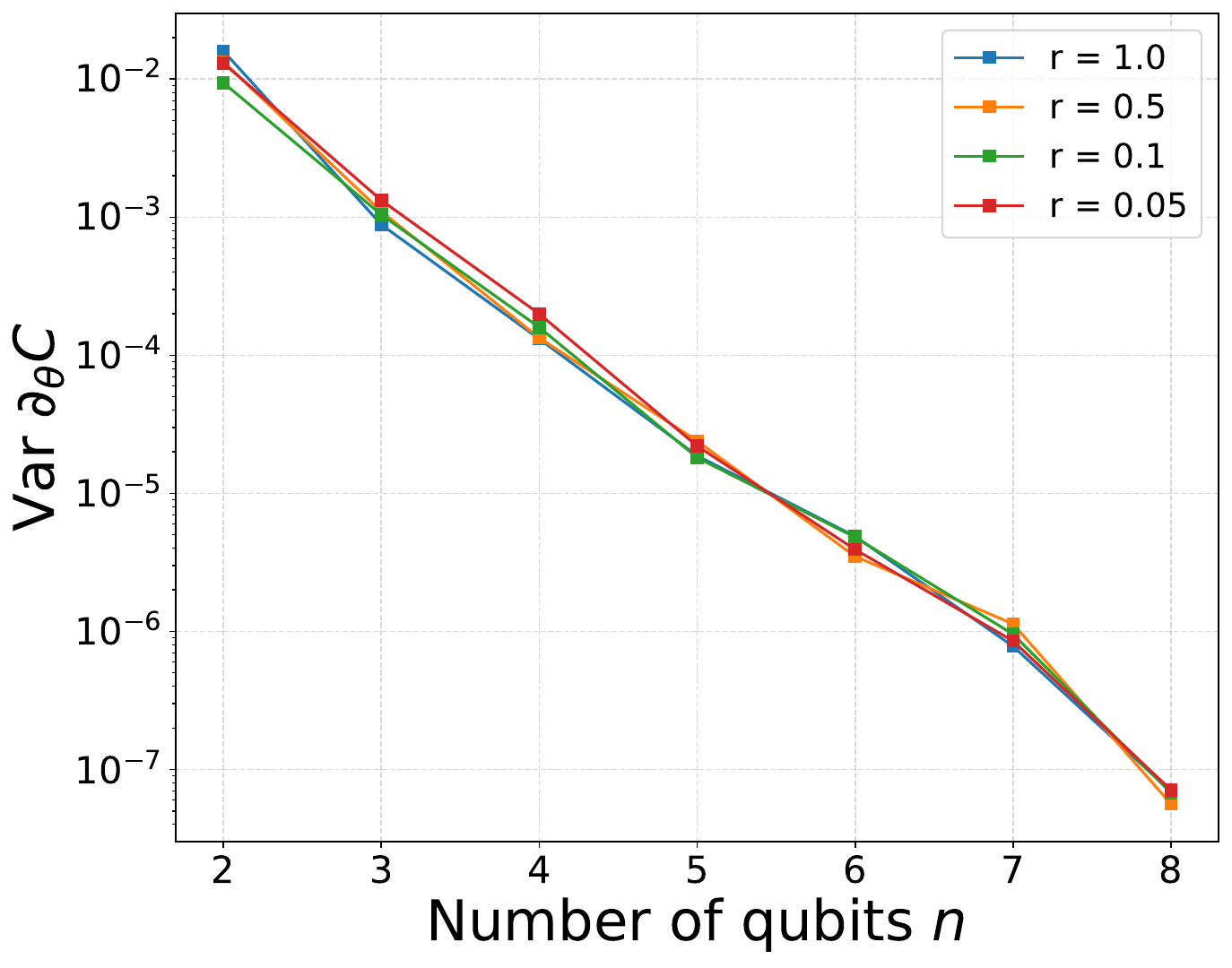}}
    \caption{Scaling of the  cost-function gradient variance with the number of qubits $n$. Rotation  angles are sampled from $[\theta_l^i, \theta_l^i + 2\pi r]$ with $\theta_l^i \sim_R {\rm Unif} [0,2\pi]$.  
    Results are shown for different levels of adversarial noise in non-local CNOT gates: $(a)$ concurrence $=1$, $(b)$ $0.9$ and $(c)$  $0.8$. 
 Gradients are computed with respect to $\theta_1^1$ (first qubit, first layer), plots are averaged over 100 trials  and $L = 10$.}
\label{fig:restricted_initialisation}
\vspace{-0.6cm}
\end{figure}
\emph{Effect of Restricted Parameter Initialization:} Fig.~\ref{fig:restricted_initialisation} studies the impact of restricted parameter initialization by sampling angles from  $[\theta_l^i,\theta_l^i+2\pi r]$, where $r \in (0,1]$. Here, $r=1$ corresponds to uniform exploration of the full parameter space, while $r \ll 1$ confines the ansatz near a fixed initialization point. 
In the noiseless case (Fig.~\ref{subfig:variance_cost_vs_num_qubit_ideal}), unrestricted initialization ($r=1)$ yields high initial variance that decays rapidly with system size $n$, while restricted initialization $(r \ll 1)$ slows this decay.
With increasing noise (Fig.~\ref{subfig:variance_cost_vs_num_qubit_conc_point_nine}-Fig.~\ref{subfig:variance_cost_vs_num_qubit_conc_point_eight}), the gradient variance decreases more rapidly, indicating earlier onset of barren plateaus. At high noise level, the variance collapses across all $r$, rendering the ansatz untrainable even at moderate system sizes.

\vspace{-0.3cm}
\section{Conclusion}
\vspace{-0.2cm}

Distributed implementations of VQAs rely critically on shared entanglement across QPUs, which is typically assumed to be secure. We show that this assumption introduces a fundamental vulnerability. We develop an adversarial framework that maps entanglement-level perturbations to gate-level noise via an explicit Kraus representation. To quantify the impact of such perturbations, we introduce Kraus expressibility, a metric that extends unitary expressibility to noisy and adversarial settings.
Our analysis establishes a principled connection between Kraus expressibility and trainability by bounding the variance of cost-function gradients. This reveals a key insight: 
by balancing purity loss and noise-induced correlations, a stealthy adversary can avoid barren plateaus, by maintaining non-vanishing gradients, while systematically biasing the optimization towards incorrect solutions.
Numerical results corroborate this behavior and show that standard mitigation strategies, such as restricted parameter initialization and local cost functions, can fail under adversarial noise.

\emph{Limitations and future directions:} Our analysis benchmarks Kraus expressibility against Haar-random unitaries. Extending this framework to more general quantum channel ensembles, including open-system dynamics, is an important direction. While we consider controlled adversarial perturbations, evaluating the feasibility of such attacks under realistic hardware constraints (e.g., decoherence, entanglement purification), and exploring whether structured ansätze provide robustness, remain open problems. Our results also highlight the need for adversarially robust training and verification protocols for distributed quantum learning.
\bibliography{references}
\bibliographystyle{unsrt}

\newpage
\appendix
\section{Noisy cat-entangler/cat-disentangler based CNOT gate} \label{app:noisy_cnot_cat}
Let us consider the CNOT implementation via the \textit{cat-entangler/cat-disentangler} protocol discussed in the previous section. In the network scenario, a central agent may prepare the Bell entangled-pair and distribute among the two QPUs that attempt to execute the \textit{nonlocal} CNOT operation. In the adversarial setting, the central agent may be compromised and under the influence of the adversary (say Eve) distribute the state $\ket{\tilde{q}_{AB}} = \sum_{i,j \in \{0,1\}} c_{ij} \ket{ij}$ among the two QPU nodes. Note that to preserve normalisation property, the coefficients of the shared state require to satisfy $\sum_{i,j \in \{0,1\}} |c_{ij}|^2 = 1$. 

With the compromised central node, the state at the start of the \textit{cat-entangler} operation is given by 
\begin{eqnarray}
    |\Psi\rangle_{in} = (|0\rangle_c|\psi_0\rangle + |1\rangle_c|\psi_1\rangle) \otimes (c_{00}|00\rangle + c_{01}|01\rangle + c_{10}|10\rangle + c_{11}|11\rangle)_{ab} \otimes |\psi_2\rangle_t 
\end{eqnarray}
where subscript $c$ and $t$ denotes the control and target qubits respectively and $ab$ denotes the shared state between the two QPUs. {\textcolor{black}{Note that in $|\Psi\rangle_{in}$ the states $\ket{\psi_0}$ and $\ket{\psi_1}$ are orthogonal in the ancilla space such that if the control node selects $\ket{\psi_0}$, it projects the control qubit non-destructively to the state $\ket{0}_c$ and vice-versa.}}  Also, in the limit of $c_{00}$ and $c_{11}$ taking value $1/\sqrt{2}$, and all other coefficients as zero we are back to the adversary-free scenario. Let us next follow the propagation of the state $|\Psi\rangle_{in}$ through the \textit{cat-entangler/cat-disentangler} protocol. Applying the \textit{CNOT} gate with control qubit $q_c$ and target qubit $q_A$ the input state $\ket{\Psi}_{in}$ transforms to
\begin{eqnarray}
&&\ket{\Psi_1} \coloneqq  \text{CNOT}_a^c~ |\Psi\rangle_{in} = |0\rangle_c|\psi_0\rangle (c_{00}|00\rangle + c_{01}|01\rangle + c_{10}|10\rangle + c_{11}|11\rangle)_{ab} \ket{\psi_2}_t \nonumber \\
&& \qquad \qquad \qquad \qquad \qquad+ |1\rangle_c|\psi_1\rangle (c_{00}|10\rangle + c_{01}|11\rangle + c_{10}|00\rangle + c_{11}|01\rangle)_{ab} \ket{\psi_2}_t.
\end{eqnarray}
Next, performing projective measurement on qubit $q_A$ in the computational basis we obtain the conditional states as 
\begin{eqnarray}
    \ket{\Psi_2} \coloneqq 
    \begin{cases}
        |0\rangle_c|\psi_0\rangle (c_{00}|0\rangle_b + c_{01}|1\rangle_b) \ket{\psi_2}_t + |1\rangle_c|\psi_1\rangle (c_{10}|0\rangle_b + c_{11}|1\rangle_b) \ket{\psi_2}_t, & \text{for outcome 0} \\
        |0\rangle_c|\psi_0\rangle (c_{11}|1\rangle_b + c_{10}|0\rangle_b) \ket{\psi_2}_t + |1\rangle_c|\psi_1\rangle (c_{01}|1\rangle_b + c_{00}|0\rangle_b) \ket{\psi_2}_t, & \text{for outcome 1}.
    \end{cases}
\end{eqnarray}
Following this, if we apply $\sigma_x$ on qubit $q_B$ when the measurement outcome is `1' we obtain the conditional state  
\begin{eqnarray}
    \ket{\Psi_3} \coloneqq 
    \begin{cases}
        |0\rangle_c|\psi_0\rangle (c_{00}|0\rangle_b + c_{01}|1\rangle_b) \ket{\psi_2}_t + |1\rangle_c|\psi_1\rangle (c_{10}|0\rangle_b + c_{11}|1\rangle_b) \ket{\psi_2}_t, & \text{for outcome 0} \\
        |0\rangle_c|\psi_0\rangle (c_{11}|0\rangle_b + c_{10}|1\rangle_b) \ket{\psi_2}_t + |1\rangle_c|\psi_1\rangle (c_{01}|0\rangle_b + c_{00}|1\rangle_b) \ket{\psi_2}_t, & \text{for outcome 1}.
    \end{cases}
\end{eqnarray}
Note here that unlike the adversary-free scenario, applying the gate sequence does not yield a deterministically ideal state for both outcomes due to the off-diagonal terms and a mismatch in the amplitudes. 

To track the error propagation, let us consider at the conditional state for the outcome `0'. Applying the CNOT gate with control qubit $q_B$ and target qubit $q_t$ and simplifying the state transforms to
\begin{eqnarray}
    \ket{\Psi_4} \coloneqq |0\rangle_b \bigg[ c_{00} |0\rangle_c|\psi_0\rangle  |\psi_2\rangle_t + c_{10} |1\rangle_c|\psi_1\rangle |\psi_2\rangle_t \bigg] + |1\rangle_b \bigg[ c_{01} |0\rangle_c|\psi_0\rangle |\bar{\psi}_2\rangle_t + c_{11} |1\rangle_c|\psi_1\rangle |\bar{\psi}_2\rangle_t \bigg].
\end{eqnarray}
Next, performing a \textit{Hadamard} operation on qubit $q_B$ and simplifying we obtain
\begin{eqnarray}
    && \ket{\Psi_5} \coloneqq \frac{1}{\sqrt{2}}|0\rangle_b \bigg[ |0\rangle_c|\psi_0\rangle(c_{00} |\psi_2\rangle_t + c_{01} |\bar{\psi}_2\rangle_t) +   
    |1\rangle_c|\psi_1\rangle(c_{10}|\psi_2\rangle_t + c_{11}|\bar{\psi}_2\rangle_t) \bigg]  \nonumber \\
    && \qquad \quad + \frac{1}{\sqrt{2}}|1\rangle_b \bigg[ |0\rangle_c|\psi_0\rangle(c_{00}|\psi_2\rangle_t - c_{01}|\bar{\psi}_2\rangle_t) +  
    |1\rangle_c|\psi_1\rangle(c_{10}|\psi_2\rangle_t - c_{11}|\bar{\psi}_2\rangle_t) \bigg].
\end{eqnarray}
Now, if we perform a projective measurement on qubit $q_B$ in the computational basis, {up to a global factor} the state $\Psi_5$ collapses to
\begin{eqnarray}
    \ket{\Psi_6} \coloneqq 
    \begin{cases}
        |0\rangle_c|\psi_0\rangle(c_{00}|\psi_2\rangle_t + c_{01}|\bar{\psi}_2\rangle_t) + |1\rangle_c|\psi_1\rangle(c_{10}|\psi_2\rangle_t + c_{11}|\bar{\psi}_2\rangle_t), & \text{for outcome 0} \\
        |0\rangle_c|\psi_0\rangle(c_{00}|\psi_2\rangle_t - c_{01}|\bar{\psi}_2\rangle_t) + |1\rangle_c|\psi_1\rangle(c_{10}|\psi_2\rangle_t - c_{11}|\bar{\psi}_2\rangle_t), & \text{for outcome 1}.
    \end{cases}
\end{eqnarray}
Following this, if we apply $\sigma_z$ on qubit $q_c$ when the measurement outcome is `1', {up to a global factor} we obtain the conditional state
\begin{eqnarray}
    \ket{\Psi_7} \coloneqq 
    \begin{cases}
        |0\rangle_c|\psi_0\rangle \bigg(c_{00}|\psi_2\rangle_t + c_{01}|\bar{\psi}_2\rangle_t\bigg) + |1\rangle_c|\psi_1\rangle \bigg(c_{10}|\psi_2\rangle_t + c_{11}|\bar{\psi}_2\rangle_t \bigg), & \text{for outcome 0} \\
        |0\rangle_c|\psi_0\rangle \bigg(c_{00}|\psi_2\rangle_t - c_{01}|\bar{\psi}_2\rangle_t \bigg) -   
        |1\rangle_c|\psi_1\rangle \bigg(c_{10}|\psi_2\rangle_t - c_{11}|\bar{\psi}_2\rangle_t \bigg) , & \text{for outcome 1}.
    \end{cases}
\end{eqnarray}
Note in $\ket{\Psi_7}$, the presence of non-zero $c_{01}$ and $c_{10}$ terms alongside imbalances between $c_{00}$ and $c_{11}$ prevents perfect disentanglement  thereby skewing the state. 

Note a similar treatment follows for the conditional state for the outcome `1' in $\ket{\Psi_3}$. The adversarial effect tweaking with the shared entangled state between $q_A$ and $q_B$ therefore creates a \textit{noisy} CNOT operation with the set of Kraus operators given by $\mathcal{K} \coloneqq \{ K_{ij}\}_{i,j \in \{0,1\}}$. The elements of $\mathcal{K}$ are
\begin{eqnarray}
    && E_{00} = \frac{1}{\sqrt{2}} \op{0}_c \otimes (c_{00} \mathbbm{1} + c_{01} \sigma_x)_t + \frac{1}{\sqrt{2}} \op{1}_c \otimes (c_{10} \mathbbm{1} + c_{11} \sigma_x)_t, \\
    && E_{01} = \frac{1}{\sqrt{2}} \op{0}_c \otimes (c_{00} \mathbbm{1} - c_{01} \sigma_x)_t - \frac{1}{\sqrt{2}} \op{1}_c \otimes (c_{10} \mathbbm{1} - c_{11} \sigma_x)_t, \\
    && E_{10} = \frac{1}{\sqrt{2}} \op{0}_c \otimes (c_{11} \mathbbm{1} + {c_{10}} \sigma_x)_t + \frac{1}{\sqrt{2}} \op{1}_c \otimes (c_{01} \mathbbm{1} + c_{00} \sigma_x)_t, \\
    && E_{11} = \frac{1}{\sqrt{2}} \op{0}_c \otimes (c_{11} \mathbbm{1} - {c_{10}} \sigma_x)_t - \frac{1}{\sqrt{2}} \op{1}_c \otimes (c_{01} \mathbbm{1} - c_{00} \sigma_x)_t.
\end{eqnarray}
Note that in the adversary-free scenario we have $c_{00} = 1/\sqrt{2} = c_{11}$ and $c_{01} = 0 = c_{10}$. In such a scenario all the Kraus operators reduce to the CNOT gate as $\frac{1}{2} \op{0}_c \otimes \mathbbm{1}_t + \frac{1}{2} \op{1}_c \otimes (\sigma_x)_t$.

\section{Generalized Expressivity of Quantum CPTP Circuits and Connection to Trainability} \label{app:quantum_channel_expressivity}

\subsection{Expressivity with Kraus operators} \label{app:Kraus_expressivity}

The Kraus expressibility of the PQCh-ansatz which quantifies the ability of noisy quantum circuits to explore the space of ideal unitary transformations can be defined using the following superoperator:
\begin{eqnarray}
    \mathcal{A}_{\mathcal{T}}^t (\cdot) \coloneqq \int_{\mathcal{U}(d)} d\mu (V) V^{\otimes t} (\cdot) (V^\dag)^{\otimes t} 
    - \int_{\mathcal{T}} d\nu(E) \sum_{k_1...k_l} (\otimes_{j=1}^t E_{k}^{(j)}) ~ (\cdot) ~ (\otimes_{j=1}^t E_{k}^{(j)\dag}),\label{eq:expressivity_operator}
\end{eqnarray}
where $d\mu(V)$ is the volume element of the Haar measure over the unitary group $\mathcal{U}(d)$ with $d=2^n$, $d\nu(E)$ is the volume element corresponding to uniform distribution over the PQCh-ensemble $\mathcal{T}$, {and $E_{k}^{(j)}$ is the k-th Kraus operator of the PQCh-ansatz for the j-th copy.} We start by analyzing the first term of Eq.~\eqref{eq:expressivity_operator} as
\begin{eqnarray}
     \int_{\mathcal{U}(d)} d\mu (V) V^{\otimes t} (\cdot) (V^\dag)^{\otimes t}:=\mathcal{G}^t_\mathcal{U} (\cdot).
\end{eqnarray} Then, for any input operator $\rho \coloneqq \rho^{\otimes t}$, the trace of the above quantity evaluates as
\begin{eqnarray}
    \Tr \mathcal{G}^t_\mathcal{U} (\rho) &&= \int_{\mathcal{U}(d)} d\mu  (V) ~\Tr [ V^{\otimes t} (\rho) (V^\dag)^{\otimes t} ] \nonumber \\
    && = \int_{\mathcal{U}(d)} d\mu  (V) ~\Tr \rho  ~ \text{(cyclic property of trace)} \nonumber \\
    && = \Tr \rho.
\end{eqnarray} Additionally, the invariance property of the Haar measure yields that for any integrable function $f(\cdot)$ and any unitary operator $V\in \mathcal{U}(d)$, the following relationship holds:
\begin{eqnarray}
    \int_U d\mu(U) f(U) = \int_U d\mu(U) f(UV) = \int_U d\mu(U) f(VU). \nonumber  
\end{eqnarray}
Using the above property, we then get that
\begin{eqnarray}
    \mathcal{G}^t_\mathcal{U} (\rho) &&= \int_{\mathcal{U}(d)} d\mu (V) V^{\otimes t} \rho (V^\dag)^{\otimes t} \\
    &&= \int_{\mathcal{U}(d)} d\mu (V) (WV)^{\otimes t} \rho ((WV)^\dag)^{\otimes t} \\
    &&= W^{\otimes t} \int_{\mathcal{U}(d)} d\mu (V) V^{\otimes t} \rho (V^\dag)^{\otimes t} (W^\dag)^{\otimes t} \\
    &&= W^{\otimes t} \mathcal{G}^t_\mathcal{U} (\rho) (W^\dag)^{\otimes t}, 
\end{eqnarray} 
thus for any unitary $W \in \mathcal{U}(d)$. 
This, in turn, implies that $(W^\dag)^{\otimes t} \mathcal{G}^t_\mathcal{U} (\rho) = \mathcal{G}^t_\mathcal{U} (\rho) (W^\dag)^{\otimes t}$, whereby we have the commutator relation $[ (W^\dag)^{\otimes t}, \mathcal{G}^t_\mathcal{U} (\rho)] = 0$, or alternatively, $[ \mathcal{G}^t_\mathcal{U} (\rho), W^{\otimes t} ] = 0$. 

{\color{black} In the following, we aim to characterize $\mathcal{G}^t_\mathcal{U}(\rho)$ such that the commutator relation $[ (W^\dag)^{\otimes t}, \mathcal{G}^t_\mathcal{U} (\rho)] = 0$ holds for any unitary $W$}. 
{\color{black} \begin{lemma} \label{lem:operator_Haar} The operator $\mathcal{G}^t_\mathcal{U}(\rho)$ takes the following form: \begin{align}
    &\mathcal{G}^t_\mathcal{U} (\rho) =  \alpha ~ \mathbbm{1} + \beta ~ \text{SWAP}, \quad \mbox{where}\\& \alpha=\frac{d \Tr \rho - \Tr [\rho ~ \text{SWAP}]}{d(d^2 - 1)}, \nonumber \\
    &\beta = \frac{d \Tr[\rho~ \text{SWAP}] - \Tr \rho}{d(d^2 - 1)}.
\end{align} \end{lemma}}
Detailed derivation can be found in Appendix~\ref{app:lemma1}.

Next, let us look at the second term of the superoperator $\mathcal{A}_{\mathcal{T}}^t (\rho)$ in \eqref{eq:expressivity_operator}: 
\begin{eqnarray}
     \mathcal{G}^t_{\mathcal{T}} (\rho^{\otimes t}) = \int_{\mathcal{T}} d\nu(E) \sum_{k_1...k_l} (\otimes_{j=1}^t E_{k}^{(j)}) ~ {\color{black}\rho^{\otimes t}} ~ (\otimes_{j=1}^t E_{k}^{(j)\dag}). 
 \end{eqnarray}
Note that when writing the above evolution the ordering of the Kraus evolution should be kept in mind. 

Subsequently, the expressibility superoperator (in the general noisy scenario) can be written as 
\begin{align}
    \mathcal{A}_{\mathcal{T}}^t (\rho^{\otimes t}) &= {\color{black}\mathcal{G}^t_\mathcal{U}(\rho^{\otimes t})- \mathcal{G}^t_{\mathcal{T}} (\rho^{\otimes t})} \\ &=\color{black}{\alpha ~ \mathbbm{1} + \beta ~ \text{SWAP}} - \color{black}{\int_{\mathcal{T}} d\nu(E) \sum_{k_1...k_l} (\otimes_{j=1}^t E_{k}^{(j)}) ~ \rho^{\otimes t} ~ (\otimes_{j=1}^t E_{k}^{(j)\dag})} \\
    & = \color{black} {\mathcal{A}^{(t)}} - \color{black}{\mathcal{M}^{(t)}}. \label{eq:superoperator_difference}
\end{align}
The following theorem then characterizes the squared-norm of the expressibility superoperator for $t=2$.
\begin{theorem} 
The  Kraus expressibility norm of the PQCh-ansatz with respect to the state $\rho$ is given by   
    \begin{eqnarray}
        || \mathcal{A}_{\mathcal{T}}^2 (\rho^{\otimes 2}) ||_2^2 = (\alpha^2 + \beta^2) ~ d^2 + 2 \alpha \beta d - 2 [\alpha + \beta \bar{\nu}] + \mathcal{N}_{\text{noise}},
    \end{eqnarray} 
    where {$\alpha = \frac{d - \Tr \rho^2}{d(d^2 - 1)}$, $\beta = \frac{d \Tr \rho^2 - 1}{d(d^2 - 1)}$, $\bar{\nu}$} is the ensemble averaged purity of the state $\rho$ after evolving through the PQCh-ansatz, and $\mathcal{N}_{\text{noise}}$ is the noise fluctuation term defined as
    \begin{eqnarray}
        \mathcal{N}_{\text{noise}} = \int_{\mathcal{T}}\int_{\mathcal{T}}d\tau(E)d\nu(F)\left[\sum_{j,k} \Tr(E_{k}\rho E_{k}^{\dagger}F_{j}\rho F_{j}^{\dagger})\right]^2
    \end{eqnarray} 
    with $\{E_k\}$ and $\{F_k\}$ representing Kraus operators corresponding to independent channel realizations sampled from 
    $\mathcal{T}$. 
\end{theorem}
\begin{proof}
We have
\begin{eqnarray*}
     || \mathcal{A}_{\mathcal{T}}^2 (\rho^{\otimes 2}) ||_2^2 && = \Tr [\mathcal{A}_{\mathcal{T}}^2 (\rho^{\otimes 2})^{\dag} ~ \mathcal{A}_{\mathcal{T}}^2 (\rho^{\otimes 2}) ] \\
    && = \Tr [(\mathcal{A}^{(2)} - \mathcal{M}^{(2)})^\dag (\mathcal{A}^{(2)} - \mathcal{M}^{(2)})] \\
    && = \Tr [\mathcal{A}^{(2)2}] + \Tr [\mathcal{M}^{(2) \dag} \mathcal{M}^{(2)}] - 2 \Re \Tr [\mathcal{A}^{(2)} M^{(2)}] 
\end{eqnarray*}
where we have used the fact that $\Tr[X] + \Tr[X^\dag] = 2 \Tr[X]$ assuming $X$ is hermitian. We next calculate each of the individual terms. We have
\begin{eqnarray*}
    \Tr [\mathcal{A}^{(2)2}] && = \Tr [(\alpha \mathbbm{1} + \beta~ \text{SWAP})(\alpha \mathbbm{1} + \beta~ \text{SWAP})] \\
    && = \alpha^2 \Tr \mathbbm{1} + \beta^2 \Tr \text{SWAP}^2 + 2 \alpha \beta \Tr\text{SWAP} \nonumber \\ 
    && = (\alpha^2 + \beta^2) ~ d^2 + 2 \alpha \beta d
\end{eqnarray*}
where we have used the fact that $\text{SWAP}^2 = \mathbbm{1}$. Next, we note that
\begin{eqnarray*}
    \Tr[\mathcal{M}^{(2)}] && = \Tr \int_{\mathcal{T}} d\nu(E) \sum_{k_1,k_2} (E_{k_1} \otimes E_{k_2}) ~ \rho^{\otimes 2} ~ (E_{k_1} \otimes E_{k_2}^\dag) \nonumber \\ 
    && = \Tr \int_{\mathcal{T}} d\nu(E) \sum_{k_1,k_2} (E_{k_1} \rho E_{k_1}^\dag) \otimes (E_{k_2} \rho E_{k_2}^\dag) \nonumber \\ 
    && = \int_{\mathcal{T}} d\nu(E) \sum_{k_1,k_2} \Tr  [(E_{k_1} \rho E_{k_1}^\dag) \otimes (E_{k_2} \rho E_{k_2}^\dag)] \nonumber \\ 
    && := \xi^{(\rho)}.
\end{eqnarray*}
Note that if $\rho$ is a density matrix, then after the Kraus evolution it remains a valid density matrix and the summation goes to identity and we have $\Tr \mathcal{M}^{(2)} = 1$. Furthermore, we have
\begin{eqnarray}
    \Tr[\mathcal{AM}] && = \Tr[(\alpha \mathbbm{1} + \beta~ \text{SWAP}) \mathcal{M}^{(2)}] \\
    && = \alpha \Tr[\mathcal{M}^{(2)}] + \beta \Tr[\text{SWAP}~\mathcal{M}^{(2)}],
\end{eqnarray}
where
\begin{eqnarray}
    \Tr [\text{SWAP} \mathcal{M}^{(2)}]  &&= \Tr \int_{\mathcal{T}} d\nu(E) ~ \text{SWAP} \sum_{k_1,k_2} (E_{k_1} \rho E_{k_1}^\dag) \otimes (E_{k_2} \rho  E_{k_2}^\dag) \nonumber \\
    && = \int_{\mathcal{T}} d\nu(E) \sum_{k_1,k_2} \Tr ~ \text{SWAP} ~ (E_{k_1} \rho E_{k_1}^\dag) \otimes (E_{k_2} \rho E_{k_2}^\dag) \nonumber \\
    && = \int_{\mathcal{T}} d\nu(E) \sum_{k_1,k_2} \Tr ~ [(E_{k_1} \rho E_{k_1}^\dag) \times (E_{k_2} \rho E_{k_2}^\dag)].
\end{eqnarray} Here, the last equality uses the fact that for two states $\rho_1$ and $\rho_2$,  we have $\Tr[\text{SWAP}~(\rho_1 \otimes \rho_2)] = \Tr[\rho_1 \rho_2]$. 

Next, noting that $\mathcal{M}^{(2)}$ is Hermitian, we have
\begin{eqnarray*}
    && \Tr[\mathcal{M}^{(2)2}] \nonumber \\
    && = \Tr[\mathcal{M}^{(2)\dag} \mathcal{M}^{(2)}] \nonumber \\
    && = \Tr \int_{\mathcal{T}} d\tau(E) \sum_{k_1,k_2} (E_{k_1} \otimes E_{k_2}) ~ \rho^{\otimes 2} ~ (E_{k_1}^\dag \otimes E_{k_2}^\dag) \int_{\mathcal{T}} d\nu(F) \sum_{j_1,j_2} (F_{j_1} \otimes F_{j_2}) ~ \rho^{\otimes 2} ~ (F_{j_1}^\dag \otimes F_{j_2}^\dag) \\
    && = \int_{\mathcal{T}} \int_{\mathcal{T}} d\tau(E) d\nu(F) \sum_{j_1,j_2} \sum_{k_1,k_2} \Tr ~ [(E_{k_1} \rho E_{k_1}^\dag \otimes E_{k_2} \rho E_{k_2}^\dag) ~ (F_{j_1} \rho F_{j_1}^\dag \otimes F_{j_2} \rho F_{j_2}^\dag)] \\
    && = \int_{\mathcal{T}} \int_{\mathcal{T}} d\tau(E) d\nu(F) \sum_{j_1,j_2} \sum_{k_1,k_2} \Tr ~ [E_{k_1} \rho E_{k_1}^\dag  F_{j_1} \rho F_{j_1}^\dag ~ \otimes ~ E_{k_2} \rho E_{k_2}^\dag F_{j_1} \rho F_{j_1}^\dag ] \\
    && = \int_{\mathcal{T}} \int_{\mathcal{T}} d\tau(E) d\nu(F) \sum_{j_1,k_1} \Tr ~ [E_{k_1} \rho E_{k_1}^\dag  F_{j_1} \rho F_{j_1}^\dag] ~ \times ~ \sum_{j_2,k_2} \Tr [E_{k_2} \rho E_{k_2}^\dag F_{j_1} \rho F_{j_1}^\dag ].
\end{eqnarray*}
Now, since the set of operators for $j_1$ and $j_2$,  and similarly for $k_1$ and $k_2$, are the same, we can simplify further as 
\begin{eqnarray}
    \Tr[\mathcal{M}^{(2)2}] = \int_{\mathcal{T}} \int_{\mathcal{T}} d\tau(E) d\nu(F) \bigg[ \sum_{j,k} \Tr ~ E_{k} \rho E_{k}^\dag  F_{j} \rho F_{j}^\dag \bigg]^2. \nonumber \\ 
\end{eqnarray}
Now, combining everything we have 
\begin{eqnarray}
    && || \mathcal{A}_{\mathcal{T}}^{t} (\rho^{\otimes 2}) ||_2^2 = (\alpha^2 + \beta^2) ~ d^2 + 2 \alpha \beta d \nonumber  - 2 \bigg[\alpha \xi^{(\rho)} + \beta \int_{\mathcal{E}} d\mu(E) \sum_{k_1,k_2} \Tr ~ E_{k_1} \rho E_{k_1}^\dag E_{k_2} \rho E_{k_2}^\dag \bigg] \nonumber \\
    && \qquad \qquad \qquad \qquad \qquad  + \int_{\mathcal{T}} \int_{\mathcal{T}} d\tau(E) d\nu(F) \bigg[ \sum_{j,k} \Tr ~ E_{k} \rho E_{k}^\dag  F_{j} \rho F_{j}^\dag \bigg]^2.
\end{eqnarray}
The expressibility of the PQCh-ansatz is then given by $\Delta_{\mathcal{T}}^\rho \coloneqq || \mathcal{A}_{\mathcal{T}}^2 (\rho^{\otimes 2}) ||_2$.
\end{proof}

\subsection{Noisy Gradient Evolution} \label{app:noisy_gradient_evolution}
Consider a PQCh-ansatz that applies a parameterized CPTP quantum channel $\mathcal{E}_{\vec{\theta}}^{\mathbf{p}}$ on an initial starting state $\rho_0$ to yield the `noisy' quantum state $\rho_{\rm out}= \mathcal{E}_{\vec{\theta}}^{\mathbf{p}}(\rho_0)$. For a  given Hamiltonian $\mathcal{H}$, we define the cost function associated with this circuit as
\begin{eqnarray}
    \mathcal{C}_{\rho,\mathcal{H}} = \Tr[\mathcal{H} ~ \mathcal{E}_{\vec{\theta}}^{\mathbf{p}}(\rho_0)].
\end{eqnarray}
Throughout, we consider the general setting when the quantum channel is a composition of CPTP channels as
\begin{align}
 \mathcal{E}_{\vec{\theta}}^{\mathbf{p}} (\rho_0) &= \Kcal_L (\mathbf{p}_L,\vec{\theta}_L) \circ \Kcal_{L-1} (\mathbf{p}_{L-1},\vec{\theta}_{L-1}) \circ ... \circ \Kcal_1 (\mathbf{p}_1,\vec{\theta}_1) ~ \rho \\
    & = \mathcal{E}_{\Lcal} \circ \mathcal{K}_k(\pbf_{k},\thetavec_{k}) \circ \mathcal{E}_{\Rcal},
\end{align} 
where $\mathcal{E}_{\Lcal}= \mathcal{K}_{L}(\pbf_L,\thetavec_L) \circ \hdots \mathcal{K}_{k+1}(\pbf_{k+1},\thetavec_{k+1})$ and 
$\mathcal{E}_{\Rcal}= \mathcal{K}_{k-1}(\pbf_{k-1},\thetavec_{k-1}) \circ \hdots \mathcal{K}_{1}(\pbf_{1},\thetavec_{1})$ denote the left and right channel compositions, respectively.
Let us denote $\rho_\Rcal = \mathcal{E}_\Rcal \circ \rho$, $\Tilde{\rho}_k = \mathcal{U}_k \circ \rho_\Rcal = e^{-i \theta_k V_k}~ \rho_\Rcal~ e^{i \theta_k V_k}$, and finally $\rho_{\text{out}} = \mathcal{E}_\Lcal \circ \Tilde{\rho}_k$.  With these definitions, we have the cost function as 
\begin{eqnarray}
    \mathcal{C}_{\rho,\mathcal{H}} &&= \mathcal{H} ~ \mathcal{E}_{\vec{\theta}}^{\mathbf{p}}(\rho_0) = \Tr [\mathcal{H} \rho_{\text{out}}] \\
    && = \Tr [\mathcal{H}~ \mathcal{E}_\Lcal \circ \Tilde{\rho}_k] \\
    && \stackrel{(a)}{=} \Tr [\mathcal{H} \sum_m E_{Lm} \Tilde{\rho}_k E_{Lm}^\dag] \\
    && = \sum_m \Tr [\mathcal{H}~ E_{Lm} \Tilde{\rho}_k E_{Lm}^\dag] \\
    && = \sum_m \Tr [E_{Lm}^\dag \mathcal{H} E_{Lm}~ \Tilde{\rho}_k ] \\
    && = \Tr [ (\sum_m E_{Lm}^\dag \mathcal{H} E_{Lm})~ \Tilde{\rho}_k ] \\
    && = \Tr [ (\sum_m E_{Lm}^\dag \mathcal{H} E_{Lm})~  e^{-i \theta_k V_k} \rho_\Rcal e^{i \theta_k V_k} ] \\
    && \stackrel{(b)}{=} \Tr [ \mathcal{H}_\Lcal ~  e^{-i \theta_k V_k} \rho_\Rcal e^{i \theta_k V_k} ], \label{eq:cost_derivation_k}
\end{eqnarray} where in $(a)$ we used the Kraus representation of the channel $\mathcal{E}_\Lcal \circ \Tilde{\rho}_k =\sum_{m}E_{Lm} \Tilde{\rho}_k E_{Lm}^{\dag}$, and in $(b)$ we used $\mathcal{H}_\Lcal=\sum_m E_{Lm}^\dag \mathcal{H} E_{Lm}$ to denote the backward evolution of the Hamiltonian till the $k^{\text{th}}$ operation. Then, the $k^{\text{th}}$ unitary gate operation is applied and measurement is performed to compute the final trace. 

The derivation in \eqref{eq:cost_derivation_k} explicitly writes out the cost function $C_{\rho,\mathcal{H}}$ as a function of the $k$-th parameter vector $\theta_k$. We now use this to calculate the partial derivative of the cost function,
\begin{eqnarray}
    \partial_k \mathcal{C}_{\rho,\mathcal{H}} &&= \partial_{\theta_k} \Tr [ \mathcal{H}_\Lcal ~  e^{-i \theta_k V_k} \rho_\Rcal e^{i \theta_k V_k} ] \\
    && = \Tr [ \mathcal{H}_\Lcal ~  \partial_{\theta_k} e^{-i \theta_k V_k} \rho_\Rcal e^{i \theta_k V_k} ] \\
    &&= \Tr [ \mathcal{H}_\Lcal ~ \Bigl( e^{-i \theta_k V_k} \rho_\Rcal ~ \partial_{\theta_k} e^{i \theta_k V_k} + \partial_{\theta_k} (e^{-i \theta_k V_k}) \rho_\Rcal   e^{i \theta_k V_k}\Bigr) ] \\
    &&= \Tr [ \mathcal{H}_\Lcal ~ \Bigl( e^{-i \theta_k V_k} \rho_\Rcal ~ i V_k e^{i \theta_k V_k}  - i V_k e^{-i \theta_k V_k} \rho_\Rcal   e^{i \theta_k V_k} \Bigr)] \\
    &&= \Tr [ i \mathcal{H}_\Lcal ~ \Bigl(  e^{-i \theta_k V_k} \rho_\Rcal ~  e^{i \theta_k V_k} V_k - V_k e^{-i \theta_k V_k} \rho_\Rcal   e^{i \theta_k V_k} \Bigr)] \\
    &&= \Tr [ i \mathcal{H}_\Lcal  [ e^{-i \theta_k V_k} \rho_\Rcal ~  e^{i \theta_k V_k}, V_k] ] \\
    &&= \Tr [ i \rho_\Rcal  [ V_k, e^{i \theta_k V_k} \mathcal{H}_\Lcal ~  e^{-i \theta_k V_k}] ] \\
    &&= \Tr [ i \rho_\Rcal  [ V_k, \Tilde{\mathcal{H}}_\Lcal (\theta_k)] ], \label{eq:partialderivative}
\end{eqnarray} where in the second last equality we have used that
$V_k$ commutes with $e^{i \theta_k V_k}$ and the last equality uses that $\tilde{\mathcal{H}}_\Lcal(\theta_k)= e^{i \theta_k V_k} \mathcal{H}_\Lcal ~  e^{-i \theta_k V_k}$.

We now analyze the quantity $\Tilde{\mathcal{H}}_\Lcal$. For Hermitian generators $V_k$ satisfying $V_k^2 = \mathbbm{1}$ we have 
\begin{eqnarray}
    \Tilde{\mathcal{H}}_\Lcal(\theta_k) &&= e^{i \theta_k V_k} \mathcal{H}_\Lcal ~  e^{-i \theta_k V_k} \\
    && = (\cos \theta_k \mathbbm{1} + i \sin \theta_k V_k) \mathcal{H}_\Lcal ~  (\cos \theta_k \mathbbm{1} - i \sin \theta_k V_k) \nonumber \\
    && = \cos^2 \theta_k \mathcal{H}_\Lcal - i \sin \theta_k \cos \theta_k \mathcal{H}_\Lcal V_k + i \sin \theta_k \cos \theta_k V_k \mathcal{H}_\Lcal + \sin^2 \theta_k V_k \mathcal{H}_\Lcal V_k \\
    && = \cos^2 \theta_k \mathcal{H}_\Lcal + \sin^2 \theta_k V_k \mathcal{H}_\Lcal V_k  + \frac{i}{2} \sin 2\theta_k [V_k, \mathcal{H}_\Lcal].
\end{eqnarray}
This yields that
\begin{eqnarray}
    V_k \Tilde{\mathcal{H}}_\Lcal(\theta_k) && = \cos^2 \theta_k V_k \mathcal{H}_\Lcal + \sin^2 \theta_k V_k^2 \mathcal{H}_\Lcal V_k + \frac{i}{2} \sin 2\theta_k (V_k^2 \mathcal{H}_\Lcal -  V_k \mathcal{H}_\Lcal V_k) \\
    && = \cos^2 \theta_k V_k \mathcal{H}_\Lcal + \sin^2 \theta_k \mathcal{H}_\Lcal V_k  + \frac{i}{2} \sin 2\theta_k (\mathcal{H}_\Lcal -  V_k \mathcal{H}_\Lcal V_k).
\end{eqnarray}
Similarly, we have
\begin{eqnarray}
    \Tilde{\mathcal{H}}_\Lcal(\theta_k) V_k && = \cos^2 \theta_k \mathcal{H}_\Lcal V_k + \sin^2 \theta_k V_k \mathcal{H}_\Lcal + \frac{i}{2} \sin 2\theta_k ( V_k \mathcal{H}_\Lcal V_k - \mathcal{H}_\Lcal).
\end{eqnarray}
We then have the commutator as
\begin{eqnarray}
    [V_k,\Tilde{\Hcal}_\Lcal(\theta_k)] &&= (cos^2 \theta_k - \sin^2 \theta_k) [V_k, \mathcal{H}_\Lcal]  + i \sin 2\theta_k (\mathcal{H}_\Lcal - V_k \mathcal{H}_\Lcal V_k).
\end{eqnarray}
Plugging this into \eqref{eq:partialderivative} yields that
\begin{eqnarray}
    \partial_k \mathcal{C}_{\rho,\mathcal{H}} &&= i \cos 2\theta_k \Tr [\rho_\Rcal [V_k,\mathcal{H}_\Lcal]] - \sin 2\theta_k \Tr[\rho_\Rcal (\mathcal{H}_\Lcal - V_k \mathcal{H}_\Lcal V_k)].
\end{eqnarray}
Now, {\color{black} for fixed left and right noisy channels $\mathcal{E}_L$ and $\mathcal{E}_R$, we evaluate the expected partial derivative of the cost function with respect to the parameter randomness} as
\begin{align}
    \mathbbm{E}_{\theta_k} [\partial_k \mathcal{C}_{\rho,\mathcal{H}}] &= \frac{1}{2\pi} \int_0^{2\pi} d\theta_k (\partial_k \mathcal{C}_{\rho,\mathcal{H}}) \nonumber \\
    & = i \Tr[\rho_\Rcal [V_k, \mathcal{H}_\Lcal]] \frac{1}{2\pi} \int_0^{2\pi} \cos 2\theta_k d\theta_k  - \frac{1}{2\pi} \Tr[\rho_\Rcal (\mathcal{H}_\Lcal - V_k \mathcal{H}_\Lcal V_k)] \int_0^{2\pi} d\theta_k \sin 2\theta_k \nonumber \\
    & = 0,
\end{align}
and the corresponding variance as 
\begin{eqnarray}
    \text{var}(\partial_k \mathcal{C}_{\rho,\mathcal{H}}) = \mathbbm{E}_{\theta_k}[(\partial_k \mathcal{C}_{\rho,\mathcal{H}})^2] - \mathbbm{E}_{\theta_k}[\partial_k \mathcal{C}_{\rho,\mathcal{H}}]^2 = \mathbbm{E}_{\theta_k}[(\partial_k \mathcal{C}_{\rho,\mathcal{H}})^2],
\end{eqnarray}
where we have 
\begin{eqnarray}
    (\partial_k \mathcal{C}_{\rho,\mathcal{H}})^2= \Tr [i \rho_\Rcal [V_k,\Tilde{\Hcal}_\Lcal(\theta_k)]]^2 = - \Tr [\rho_\Rcal [V_k,\Tilde{\Hcal}_\Lcal(\theta_k)]]^2 = - \Tr [\rho_\Rcal^{\otimes 2} ~[V_k,\Tilde{\Hcal}_\Lcal(\theta_k)]^{\otimes 2}].
\end{eqnarray}
Thus, the variance with respect to the parameter $\theta_k$ {\color{black} (for some fixed $\mathcal{E}_L$ and $\mathcal{E}_R$)} is given by
\begin{eqnarray}
    \text{var}(\partial_k \mathcal{C}_{\rho,\mathcal{H}}) &&= - \mathbbm{E}_{\theta_k} [\Tr [\rho_\Rcal^{\otimes 2} ~[V_k,\Tilde{\Hcal}_\Lcal(\theta_k)]^{\otimes 2}]].\label{eq:final_variance}
\end{eqnarray}

We now evaluate the expected variance of the partial derivative with respect to the noisy channel realizations. To this end, we take the expectation over the distributions of the left and right CPTP maps $\mathcal{E}_\Lcal$ and $\mathcal{E}_\Rcal$ as
\begin{eqnarray}
    \mathbb{E}_{\mathcal{E}_\Lcal, \mathcal{E}_\Rcal}[\text{var}(\partial_k \mathcal{C}_{\rho,\mathcal{H}})] &&= -\mathbb{E}_{\mathcal{E}_\Lcal, \mathcal{E}_\Rcal}[\mathbb{E}_{\theta_k}[\text{Tr}[\rho_\Rcal^{\otimes 2}[V_k, \tilde{\mathcal{H}}_\Lcal(\theta_k)]^{\otimes 2}]]] \nonumber \\
    && = -\text{Tr}[\mathbb{E}_{\mathcal{E}_\Rcal}[\rho_\Rcal^{\otimes 2}] \mathbb{E}_{\mathcal{E}_\Lcal}[\mathbb{E}_{\theta_k}[[V_k, \tilde{\mathcal{H}}_\Lcal(\theta_k)]^{\otimes 2}]]] \label{eq:expectedvariance}.
\end{eqnarray}
Let us now define the parameter-integrated observable as $\Omega_k \equiv \mathbb{E}_{\theta_k}[[V_k, \tilde{\mathcal{H}}_L(\theta_k)]^{\otimes 2}]$ and calculate the first quantity as
\begin{eqnarray}
    \mathbb{E}_{\mathcal{E}_\Rcal}[\rho_\Rcal^{\otimes 2}] &&= \int_{\mathcal{E}_\Rcal} d\nu(\mathcal{E}_\Rcal) \sum_{j_1, j_2} (E^{(R)}_{ j_1} \otimes E^{(R)}_{j_2}) \rho^{\otimes 2} (E^{(R)\dagger}_{j_1} \otimes E^{(R)\dagger}_{j_2}) \nonumber \\
    && = \mathcal{M}^{(2)}_\Rcal (\rho^{\otimes 2}), 
\end{eqnarray} where recall from the definition of expressivity superoperator in \eqref{eq:superoperator_difference} that
\begin{eqnarray}
    &&  \mathcal{M}^{(2)}_\Rcal(\rho^{\otimes 2}) = \mathcal{A}^{(2)}(\rho^{\otimes 2}) - \mathcal{A}_{\mathcal{T_R}}^{(2)}(\rho^{\otimes 2}).
\end{eqnarray} 

Hence, the expected variance in \eqref{eq:expectedvariance} evaluates as
\begin{eqnarray}
\mathbb{E}_{\mathcal{E}_\Lcal, \mathcal{E}_\Rcal}[\text{var}(\partial_k \mathcal{C}_{\rho, \mathcal{H}})]
    &&= -\text{Tr}[\mathcal{M}^{(2)}_\Rcal(\rho^{\otimes 2}) \mathbb{E}_{\mathcal{E}_\Lcal}[\Omega_k]] \\
    &&= -\text{Tr}[(\mathcal{A}^{(2)}(\rho^{\otimes 2}) - \mathcal{A}_{\mathcal{T_R}}^{(2)}(\rho^{\otimes 2})) \mathbb{E}_{\mathcal{E}_\Lcal}[\Omega_k]] \\
    && = -\text{Tr}[\mathcal{A}^{(2)}(\rho^{\otimes 2}) \mathbb{E}_{\mathcal{E}_\Lcal}[\Omega_k]] + \text{Tr}[\mathcal{A}_{\mathcal{T_R}}^{(2)}(\rho^{\otimes 2}) \mathbb{E}_{\mathcal{E}_\Lcal}[\Omega_k]] \nonumber \\
    && = \text{var}_\Rcal(\partial_k \mathcal{C}_{\rho, \mathcal{H}}) + \text{Tr}[\mathcal{A}_{\mathcal{T_R}}^{(2)}(\rho^{\otimes 2}) \mathbb{E}_{\mathcal{E}_\Lcal}[\Omega_k]], 
\end{eqnarray} 
where we have explicitly defined $$\text{var}_\Rcal(\partial_k \mathcal{C}_{\rho, \mathcal{H}})=-\text{Tr}[\mathcal{A}^{(2)}(\rho^{\otimes 2}) \mathbb{E}_{\mathcal{E}_\Lcal}[\Omega_k]] $$ 
as the expected variance of the partial-derivative of the cost function over the left CPTP maps when the right sub-circuit forms a unitary 2-design. We then get that
\begin{align}
 |\mathbb{E}_{\mathcal{E}_\Lcal, \mathcal{E}_\Rcal}[\text{var}(\partial_k \mathcal{C}_{\rho, \mathcal{H}})] - \text{var}_\Rcal(\partial_k \mathcal{C}_{\rho, \mathcal{H}})| &\le |\text{Tr}[\mathcal{A}_{\mathcal{T_R}}^{(2)}(\rho^{\otimes 2}) \mathbb{E}_{\mathcal{E}_\Lcal}[\Omega_k]]| \\
    & \le ||\mathcal{A}_{\mathcal{T_R}}^{(2)}(\rho^{\otimes 2})||_2 \mathbb{E}_{\mathcal{E}_\Lcal}[||\Omega_k||_2], \label{eq:main_equation}
\end{align}where the last inequality follows from Jensen's inequality and an application of the Cauchy-Schwartz inequality. 
Recalling the definition of $\Omega_k=\mathbb{E}_{\theta_k}[[V_k, \tilde{\mathcal{H}}_\Lcal(\theta_k)]^{\otimes 2}]$, we use  the triangle inequality to get that
\begin{eqnarray}
    ||\Omega_k||_2 \le \mathbb{E}_{\theta_k}[||[V_k, \tilde{\mathcal{H}}_\Lcal(\theta_k)]^{\otimes 2}||_2],
\end{eqnarray}
where we have $\tilde{\mathcal{H}}_\Lcal(\theta_k) = e^{i\theta_k V_k}\mathcal{H}_\Lcal e^{-i\theta_k V_k}$. Since $V_k$ commutes with $e^{i\theta_k V_k}$, we can write
\begin{eqnarray}
    [V_k, \tilde{\mathcal{H}}_\Lcal(\theta_k)] = e^{i\theta_k V_k}[V_k, \mathcal{H}_\Lcal]e^{-i\theta_k V_k}.
\end{eqnarray} Additionally, the Frobenius norm is unitarily invariant, which then yields that $||[V_k, \tilde{\mathcal{H}}_\Lcal(\theta_k)]^{\otimes 2}||_2 = ||[V_k, \mathcal{H}_\Lcal]^{\otimes 2}||_2 = ||[V_k, \mathcal{H}_\Lcal]||_2^2$ is independent of $\theta_k$. Together, we get that
\begin{align}
   ||\Omega_k||_2 &\le \Vert[V_k, \mathcal{H}_\Lcal] \Vert_2^2  \\&= \text{Tr}[[V_k, \mathcal{H}_\Lcal] [V_k, \mathcal{H}_\Lcal]^{\dag}]|\\
    & = 2\text{Tr}[\mathcal{H}_\Lcal^2] - 2\text{Tr}[V_k \mathcal{H}_\Lcal V_k \mathcal{H}_\Lcal], \label{eq:relation_2}
\end{align} where the latter follows by noting that $V_k$ and $\mathcal{H}_\Lcal$ are Hermitian operators with $V_k^2=\mathbbm 1$.
\begin{note}
   Consider the  Hermitian operators $V_k$ and $\mathcal{H}_\Lcal$ satisfying $V_k^2=\mathbbm 1$.  Let $A$ be the anticommutator of these operators given by $A = V_k \mathcal{H}_\Lcal + \mathcal{H}_\Lcal V_k$. Then, 
    \begin{align}
        \Tr[A^2] &= \Tr[V_k \mathcal{H}_\Lcal V_k \mathcal{H}_\Lcal + V_k \mathcal{H}_\Lcal^2 V_k + \mathcal{H}_\Lcal V_k^2 \mathcal{H}_\Lcal + \mathcal{H}_\Lcal V_k \mathcal{H}_\Lcal V_k]\nonumber \\ &= 2 \Tr[V_k \mathcal{H}_\Lcal V_k \mathcal{H}_\Lcal] + 2 \Tr[\mathcal{H}_\Lcal^2] \geq 0,
    \end{align} which follows by noting that
   $V_k^2 = \mathbbm{1}$. 
    This then implies that \begin{align}\Tr[V_k \mathcal{H}_\Lcal V_k \mathcal{H}_\Lcal] \geq - \Tr[\mathcal{H}_\Lcal^2].\label{eq:relation}\end{align}
\end{note}

Combining \eqref{eq:relation_2} and \eqref{eq:relation}, and pluggin it in \eqref{eq:main_equation} then yields that
\begin{align}
|\mathbb{E}_{\mathcal{E}_\Lcal, \mathcal{E}_\Rcal}[\text{var}(\partial_k \mathcal{C}_{\rho, \mathcal{H}})] - \text{var}_\Rcal(\partial_k \mathcal{C}_{\rho, \mathcal{H}})| 
    & \le ||\mathcal{A}_{\mathcal{T_R}}^{(2)}(\rho^{\otimes 2})||_2 \mathbb{E}_{\mathcal{E}_\Lcal}[||\Omega_k||_2] \\
    & \le 4||\mathcal{A}_{\mathcal{T_R}}^{(2)}(\rho^{\otimes 2})||_2 \mathbb{E}_{\mathcal{E}_\Lcal}[||\mathcal{H}_\Lcal||_2^2] \\&=4||\mathcal{A}_{\mathcal{T_R}}^{(2)}(\rho^{\otimes 2})||_2 \int_{\mathcal{E}_\Lcal} d\nu(\mathcal{E}_\Lcal) ||\mathcal{E}_\Lcal^\dagger(\mathcal{H})||_2^2,
\end{align}where the last equality follows by noting that since observables evolve under the adjoint map, we can write $\mathcal{H}_\Lcal = \mathcal{E}_\Lcal^\dagger(\mathcal{H})$.

\subsection{Proof of Lemma~\ref{lem:operator_Haar}}\label{app:lemma1}
We start by noting that any unitary operator $W$ can be written as $W = \exp(i \epsilon H)$ in terms of some Hermitian operator $H$. This then yields  that
\begin{eqnarray}
    W^{\otimes2} && = e^{i\epsilon H} \otimes e^{i\epsilon H} =  (e^{i\epsilon H} \otimes \mathbbm{1}) (\mathbbm{1} \otimes e^{i\epsilon H}) \nonumber \\
    && = e^{i\epsilon (H\otimes \mathbbm{1})} e^{i\epsilon (\mathbbm{1} \otimes H)} \nonumber \\
&& = e^{i\epsilon (H\otimes \mathbbm{1})} e^{i\epsilon (\mathbbm{1} \otimes H)} e^{\frac{1}{2} i \epsilon [H \otimes \mathbbm{1}, \mathbbm{1} \otimes H]} \times ~\text{(higher order terms)} \nonumber \\
    && = e^{i \epsilon \textcolor{black}{[H \otimes \mathbbm{1} + \mathbbm{1} \otimes H]}} = e^{i \epsilon \textcolor{black}{\eta}}.
\end{eqnarray}
Now, consider an operator
 $X$ satisfying the commutator relation $[X,W^{\otimes 2}] = 0$. This implies that
\begin{eqnarray}
    X~W^{\otimes 2} - W^{\otimes 2} X = 0~\forall \epsilon.
\end{eqnarray}
Now since it holds for all $\epsilon$, we have $ \partial_{\epsilon} ( XW^{\otimes 2} - W^{\otimes 2} X ) = 0$. This then implies
\begin{eqnarray}
     && X \partial_\epsilon W^{\otimes 2} - (\partial_\epsilon W^{\otimes 2}) X = 0 \\
      \implies && X i \eta W^{\otimes 2} - i \eta W^{\otimes 2} X = 0 \\
      \implies && X i \eta W^{\otimes 2} - i \eta X W^{\otimes 2} = 0 ~~~ \text{[ from (25) ]} \\
      \implies && (X \eta - \eta X) W^{\otimes 2} = 0 \\
      \implies && [X,\eta] = 0.
\end{eqnarray}
Plugging back the definition of $\eta$, we then get that $[X,H \otimes \mathbbm{1} + \mathbbm{1} \otimes H] = 0.$

\begin{remark}
    By property of kronecker product we have,
    \begin{eqnarray}
        && ((A+B) \otimes \mathbbm{1}) + (\mathbbm{1} \otimes (A+B)) \nonumber \\ 
        && = A \otimes \mathbbm{1} + B \otimes \mathbbm{1} + \mathbbm{1} \otimes A + \mathbbm{1} \otimes B \\
        && = (A \otimes \mathbbm{1} + \mathbbm{1} \otimes A) + (B \otimes \mathbbm{1} + \mathbbm{1} \otimes B).
    \end{eqnarray}
\end{remark}

We note that any hermitian matrix can be written as
\begin{eqnarray}
    H = \sum_a h_{aa} \mathcal{D}_a + \sum_{ab} f_{ab} \mathcal{R}_{ab} + \sum_{ab} g_{ab} \mathcal{I}_{ab}
\end{eqnarray}
where $\mathcal{D}_a$ is the diagonal component given by $\op{a}$, $\mathcal{R}_{ab}$ is the real component given by $\op{a}{b} + \op{b}{a}$, and $\mathcal{I}_{ab}$ is the imaginary component given by $i \op{a}{b} - i \op{b}{a}$. 
Let's first compute the diagonal term of $H \otimes \mathbbm{1} + \mathbbm{1} \otimes H$. We have the diagonal operator as
\begin{eqnarray}
    D_a = \op{a} \otimes \mathbbm{1} + \mathbbm{1} \otimes \op{a}
\end{eqnarray}
Now, the diagonal operator acts on some quantum state  as
\begin{eqnarray}
    D_a \ket{qs} && = (\op{a} \otimes \mathbbm{1} + \mathbbm{1} \otimes \op{a}) \ket{qs} \\
    && = \op{a} \ket{q} \otimes \ket{s} + \ket{q} \otimes \op{a}\ket{s} \\
    && = \delta_{aq} \ket{as} + \delta{as} \ket{qa} \\
    && = \delta_{aq} \ket{qs} + \delta{as} \ket{qs} \\
    && = (\delta_{aq} + \delta_{as}) \ket{qs}
\end{eqnarray}
Now, consider the expectation value of the commutator for any arbitrary state. We then have 
\begin{eqnarray}
    && \bra{pr} [X,D_a] \ket{qs} = 0 \\
    \implies && \bra{pr} XD_a - D_a X \ket{qs} = 0 \\
    \implies && \bra{pr} XD_a \ket{qs} - \bra{pr} D_a X \ket{qs} = 0 \\ 
    \implies && \bra{pr} X (\delta_{aq} + \delta_{as}) \ket{qs} - \bra{pr} (\delta_{ap} + \delta_{ar}) X \ket{qs} = 0 \nonumber \\ \\
    \implies && (\delta_{aq} + \delta_{as}) X_{prqs} - (\delta_{ap} + \delta_{ar}) X_{prqs} = 0 \\
    \implies && (\delta_{aq} + \delta_{as} -\delta_{ap} - \delta_{ar}) X_{prqs} = 0 \label{eq:cond_1}
\end{eqnarray}
Assuming that $X_{prqs} \neq 0$ in general, the above relation implies that
\begin{eqnarray}
    && \delta_{aq} + \delta_{as} -\delta_{ap} - \delta_{ar} = 0,
\end{eqnarray} which leads to the following two choices that satisfy the above relation:
\begin{eqnarray}
    && q=p ~ \text{and} ~ s=r \\
    \text{or} ~~ && q=r ~ \text{and} ~ s=p. \label{eq:conditions}
\end{eqnarray}The former choice yields $X_{prqs}=X_{pr}\delta_{pq} \delta_{rs}$ as a feasible operator definition satisfying \eqref{eq:cond_1}, while the latter choice yields  $X_{prqs}=X_{pr}\delta_{ps} \delta_{qr}$.

Next look at the imaginary component. We have the operator 
\begin{eqnarray}
    \mathcal{I}_{ab} = i \op{a}{b} - i \op{b}{a} = i (\op{a}{b} - \op{b}{a}).
\end{eqnarray}
The effect of this operator on the state is as
\begin{eqnarray}
    \mathcal{I}_{ab} \ket{k} = i \op{a}{b} \ket{k} - i \op{b}{a} \ket{k} = i \delta_{bk} \ket{a} - i \delta_{ak} \ket{b} 
\end{eqnarray}
and similarly
\begin{eqnarray}
    \bra{j} \mathcal{I}_{ab}  = i \delta_{ja} \bra{b} - i \delta_{jb} \bra{a} 
\end{eqnarray}
Now, let $\gamma_{ab} = \mathcal{I}_{ab} \otimes \mathbbm{1} + \mathbbm{1} \otimes \mathcal{I}_{ab}$. The effect of this operator on some arbitrary state is given by
\begin{eqnarray}
    && \gamma_{ab} \ket{qs} \nonumber \\
    && = (\mathcal{I}_{ab} \otimes \mathbbm{1} + \mathbbm{1} \otimes \mathcal{I}_{ab}) \ket{qs} \\
    && = i ( \delta_{bq}  \ket{a} -  \delta_{aq} \ket{b}) \ket{s} + i \ket{q} ( \delta_{bs}  \ket{a} -  \delta_{as} \ket{b})  \\
    && = i ( \delta_{bq}  \ket{as} -  \delta_{aq} \ket{bs} + \delta_{bs}  \ket{qa} -  \delta_{as} \ket{qb})
\end{eqnarray}
Similarly, we have
\begin{eqnarray}
    \bra{pr} \gamma_{ab} = i ( \delta_{pa}  \bra{br} -  \delta_{pb} \bra{ar} + \delta_{ra}  \bra{pb} -  \delta_{rb} \bra{pa})
\end{eqnarray}
Next, considering the expectation value of the commutator for arbitrary state we have 
\begin{eqnarray}
    && \bra{pr} (X\gamma_{ab} - \gamma_{ab} X) \ket{qs} = 0 \\
    \implies && \bra{pr} (X\gamma_{ab} \ket{qs} - \bra{pr} \gamma_{ab} X) \ket{qs} = 0 \\
    \implies && i \bra{pr} X ( \delta_{bq}  \ket{as} -  \delta_{aq} \ket{bs} + \delta_{bs}  \ket{qa} -  \delta_{as} \ket{qb}) \nonumber \\
    &&- i ( \delta_{pa}  \bra{br} -  \delta_{pb} \bra{ar} + \delta_{ra}  \bra{pb} -  \delta_{rb} \bra{pa}) X \ket{qs} = 0 \nonumber \\ \\
    \implies && ( \delta_{bq}  X_{pras} -  \delta_{aq} X_{prbs} + \delta_{bs}  X_{prqa} -  \delta_{as} X_{prqb}) \nonumber \\
    && - ( \delta_{pa}  X_{brqs} -  \delta_{pb} X_{arqs} + \delta_{ra}  X_{pbqs} -  \delta_{rb} X_{paqs}) = 0 \label{eq:61}
\end{eqnarray}
Now, from the analysis of the diagonal components in \eqref{eq:conditions}, we have that either $\{p,r\} = \{q,s\}$, whereby the operator takes the form $X_{prqs}= X_{pr} \delta_{pq} \delta_{rs}$, or  $\{p,r\} = \{s,q\}$, whereby the operator takes the form $ X_{prqs} = X_{pr} \delta_{ps} \delta_{qr}.$
Consequently, in the first case, i.e., $\{p,r\} = \{q,s\}$, we have \eqref{eq:61} taking the form
\begin{eqnarray}
    && ( \delta_{bq} \delta_{pa} \delta_{rs}  X_{pr} -  \delta_{aq} \delta_{pb} \delta_{rs} X_{pr} \nonumber \\
    && \qquad \qquad + \delta_{bs} \delta_{pq} \delta_{ra}  X_{pr} -  \delta_{as} \delta_{pq} \delta_{rb} X_{pr}) \nonumber \\
    && - ( \delta_{pa} \delta_{bq} \delta_{rs}  X_{br} -  \delta_{pb} \delta_{aq} \delta_{rs} X_{ar} \nonumber \\
    && \qquad \qquad + \delta_{ra} \delta_{pq} \delta_{bs}  X_{pb} -  \delta_{rb} \delta_{pq} \delta_{as} X_{pa}) = 0. 
\end{eqnarray}
Rearranging the above relation then yields
\begin{eqnarray}
    && \delta_{bq} \delta_{pa} \delta_{rs}  (X_{pr} - X_{br}) - \delta_{aq} \delta_{pb} \delta_{rs} (X_{pr} - X_{ar}) \nonumber \\
    && + \delta_{bs} \delta_{pq} \delta_{ra}  (X_{pr} - X_{pb}) - \delta_{as} \delta_{pq} \delta_{rb} (X_{pr} - X_{pa}) = 0 \nonumber \\
\end{eqnarray}
Assuming that the delta functions are non-zero implies that 
\begin{eqnarray}
    && X_{pr} = X_{br}, ~~ X_{pr} = X_{ar} \\
    && X_{pr} = X_{pb}, ~~ X_{pr} = X_{pa}. 
\end{eqnarray}
This implies $X_{ar} = X_{br} = \alpha$ and $X_{pb} = X_{pa} = \alpha$. 
We then have the expectation value as $X_1 = \bra{pr} X \ket{qs} = \alpha \delta_{pq} \delta_{rs} = \alpha ~ \mathbbm{1}$.

Next, let's look at the second case, i.e., $\{p,r\} = \{s,q\}$. From \eqref{eq:61}, we have
\begin{eqnarray}
    && ( \delta_{bq} \delta_{ps} \delta_{ra}  X_{pr} -  \delta_{aq} \delta_{ps} \delta_{rb} X_{pr} \nonumber \\
    && \qquad \qquad + \delta_{bs} \delta_{pa} \delta_{rq}  X_{pr} -  \delta_{as} \delta_{pb} \delta_{rq} X_{pr}) \nonumber \\
    && - ( \delta_{pa} \delta_{bs} \delta_{rq}  X_{br} -  \delta_{pb} \delta_{as} \delta_{rq} X_{ar} \nonumber \\
    && \qquad \qquad + \delta_{ra} \delta_{ps} \delta_{bq}  X_{pb} -  \delta_{rb} \delta_{ps} \delta_{aq} X_{pa}) = 0 
\end{eqnarray}

On simplifying we have
\begin{eqnarray}
    && \delta_{bq} \delta_{ps} \delta_{ra}  (X_{pr} - X_{pb}) - \delta_{aq} \delta_{ps} \delta_{rb} (X_{pr} - X_{pa}) \nonumber \\
    && + \delta_{bs} \delta_{pa} \delta_{rq}  (X_{pr} - X_{br}) - \delta_{as} \delta_{pb} \delta_{rq} (X_{pr} - X_{ar}) = 0 \nonumber \\
\end{eqnarray}
Assuming that the delta functions are non-zero implies that 
\begin{eqnarray}
    && X_{pr} = X_{pb}, ~~ X_{pr} = X_{pa} \\
    && X_{pr} = X_{br}, ~~ X_{pr} = X_{ar}. 
\end{eqnarray}
This implies $X_{pb} = X_{pa} = \beta$ and $X_{br} = X_{ar} = \beta$. 
We then have the expectation value as $X_2 = \bra{pr} X \ket{qs} = \beta \delta_{ps} \delta_{rq} = \beta ~ \text{SWAP}$. {\color{black} We thus get that there exists an operator $X=\alpha \mathbbm{1}+ \beta {\rm SWAP}$ such that the commutor relation $[X,W^{\otimes 2}]=0$ satisfies.}

Recall that we had $[ \mathcal{G}^t_\mathcal{U} (\rho), W^{\otimes t} ] = 0$. 
So we can write the operator $\mathcal{G}^t_\mathcal{U} (\rho)$ as
\begin{eqnarray}
    \mathcal{G}^t_\mathcal{U} (\rho) = \alpha ~ \mathbbm{1} + \beta ~ \text{SWAP}, \label{eq:e^2_rho}
\end{eqnarray} for some $\alpha$ and $\beta$. To find these values, we first note that  
$
    \Tr \mathcal{G}^t_\mathcal{U} (\rho) = \Tr ~ \rho 
$, which implies that 
\begin{eqnarray}
    && \Tr (\alpha~\mathbbm{1} + \beta ~ \text{SWAP}) = \alpha d^2 + \beta d = \Tr \rho. \label{eq:sol_1}
\end{eqnarray}

Furthermore, we have that
\begin{eqnarray}
    \Tr ( \mathcal{G}^t_\mathcal{U} (\rho) ~ \text{SWAP}) && = \int_{\mathcal{U}(d)} d\mu (V) \Tr [ V^{\otimes t} \rho (V^\dag)^{\otimes t} ~ \text{SWAP} ] \nonumber \\
    && = \int_{\mathcal{U}(d)} d\mu(V) \Tr [ V^{\otimes t} \rho ~ \text{SWAP} ~ (V^\dag)^{\otimes t}] \nonumber \\ && =\Tr [\rho \text{SWAP}], 
\end{eqnarray}
\textcolor{black}{where we assumed that $(V^\dag)^{\otimes t}$ commutes with the SWAP operator and applied the cyclic property of trace}. Combining this with \eqref{eq:e^2_rho}, we get that  
\begin{align}
 \Tr ( \mathcal{G}^t_\mathcal{U} (\rho) ~ \text{SWAP}) & =\Tr [\alpha \mathbbm{1} ~ \text{SWAP} + \beta ~ \text{SWAP}^2] \\&  =\alpha d + \beta d^2 = \Tr [\rho ~ \text{SWAP}]. \label{eq:sol_2}
\end{align}
Solving \eqref{eq:sol_1} and \eqref{eq:sol_2}, we have
\begin{eqnarray}
    && \alpha = \frac{d \Tr \rho - \Tr [\rho ~ \text{SWAP}]}{d(d^2 - 1)}, \\
    && \beta = \frac{d \Tr[\rho~ \text{SWAP}] - \Tr \rho}{d(d^2 - 1)},
\end{eqnarray}
which in turn yields that
\begin{eqnarray}
    && \mathcal{G}^t_\mathcal{U} (\rho) = [\frac{d \Tr \rho - \Tr [\rho ~ \text{SWAP}]}{d(d^2 - 1)}] ~ \mathbbm{1} + [\frac{d \Tr[\rho~ \text{SWAP}] - \Tr \rho}{d(d^2 - 1)}] ~ \text{SWAP}.
\end{eqnarray}
\end{document}